\documentclass[aps,a4paper,pra,floatfix,amsmath,superscriptaddress,twocolumn]{revtex4-1}
\pdfoutput=1
\usepackage{amssymb}
\usepackage{graphicx}
\usepackage{graphics}
\usepackage{amsmath}
\usepackage{amsthm}
\usepackage{color}
\usepackage{comment}
\usepackage{url}
\usepackage{hyperref}

\usepackage[yyyymmdd,hhmmss]{datetime}
\usepackage{bbm}
\usepackage[utf8]{inputenc}

\newcommand{\bea}{\begin{eqnarray}}
\newcommand{\eea}{\end{eqnarray}}

\def\bi{\begin{itemize}}
\def\ei{\end{itemize}}
\def\bc{\begin{center}}
\def\ec{\end{center}}

\def\C{\ensuremath{\mathbbm{C}}}
\def\R{\ensuremath{\mathbbm{R}}}

\newcommand{\id}{\ensuremath{\mathbbm{1}}} 
\newcommand{\one}{\id}
\newcommand{\kb}[2]{\ket{#1}\bra{#2}}

\newcommand{\SO}{{\cal SO}} 
\renewcommand{\O}{{\cal O}} 

\newtheorem{theorem}{Theorem}
\newtheorem{corollary}[theorem]{Corollary}
\newtheorem{lemma}[theorem]{Lemma}

\newtheorem{observation}[theorem]{Observation}
\newcommand{\ket}[1]{\ensuremath{\left| #1\right>}}
\newcommand{\bra}[1]{\ensuremath{\left< #1\right|}}
\newcommand{\proj}[1]{\ensuremath{\ketbra{#1}{#1}}}
\newcommand{\ketbra}[2]{\ensuremath{\vert\!\!\;#1\rangle \! \langle #2\vert}}
\newcommand{\Eqref}[1]{Eq.~(\ref{#1})}                                

\newcommand{\Pf}{\mathrm{Pf}} 
\newcommand{\rank}{\mathrm{rank}}
\newcommand{\tr}{\mathrm{tr}}
\def\cE{{\cal E}}
\def\cP{{\cal P}}

\def\cS{{\cal S}}
\def\Proof{\textsc{Proof}}

\usepackage[normalem]{ulem}
\definecolor{gray}{rgb}{.6,.6,.6}

\begin{document}
 \author{C. Spee}
\affiliation{Naturwissenschaftlich-Technische Fakult\"at, Universit\"at Siegen, 57068 Siegen, Germany}
\author{K. Schwaiger}
\affiliation{Institute for Theoretical Physics, University of
Innsbruck, 6020 Innsbruck, Austria}
\author{G. Giedke}
\affiliation{Donostia International Physics Center, 20018 San Sebasti\'an, Spain}
\affiliation{Ikerbasque Foundation for Science, 48013 Bilbao, Spain}
\author{B. Kraus}
\affiliation{Institute for Theoretical Physics, University of
Innsbruck, 6020 Innsbruck, Austria}

\title{Mode-entanglement of Gaussian fermionic states}

\date{\today}

\begin{abstract}
 We investigate the entanglement of $n$-mode $n$-partite Gaussian
  fermionic states (GFS). First, we identify a reasonable definition
  of separability for GFS and derive a standard form for mixed states,
  to which any state can be mapped via \textit{Gaussian local
    unitaries} (GLU). As the standard form is unique two GFS are
  equivalent under GLU if and only if their standard forms
  coincide. Then, we investigate the important class of \textit{ local
    operations assisted by classical communication} (LOCC). These are
  central in entanglement theory as they allow to partially order the
  entanglement contained in states. We show, however, that there are no non-trivial
  Gaussian LOCC (GLOCC). That is, any GLOCC transformation can be
  accomplished via GLUs. To still obtain insights into the various
  entanglement properties of $n$-mode $n$-partite GFS we investigate
  the richer class of Gaussian \textit{stochastic} LOCC. We
  characterize Gaussian SLOCC classes of pure states and derive them
  explicitly for few-mode states. Furthermore, we consider certain
  fermionic LOCC and show how to identify the \textit{maximally
    entangled set} (MES) of pure $n$-mode $n$-partite GFS, i.e., the
  minimal set of states having the property that any other state can
  be obtained from one state inside this set via fermionic LOCC. We generalize these findings also to the pure $m$-mode $n$-partite (for $m>n$) case.
\end{abstract}

\maketitle

\section{Introduction}

Entanglement \cite{HHHH09} plays a crucial role in understanding the quantum physics
of systems composed of many subsystems or many particles. It is the
primary resource of many applications in quantum computation and
communication, and is the basis of many of the intriguing effects of
quantum many-body physics.

In multipartite systems there are various qualitatively different kinds
of entanglement. Relating them to physical properties \cite{AFOV08} or to
performable tasks \cite{BBC+93, ADHW06}, contributes to elucidating the role of entanglement
in nature and as a resource for quantum technologies \cite{acin2017}. 

One very successful approach to identify different classes of
entanglement is to consider whether states can be converted into each
other using some naturally restricted set of quantum operations,
defining states for which such conversion is mutually impossible to
belong to distinct classes. This has lead to the discovery of
inequivalent kinds of entanglement \cite{DVC00, FrankSLOCC4}, and to their classification \cite{GoWa13, eltschka2012}. Furthermore, the maximally entangled states and sets, which are the most relevant states regarding local state transformations, have been identified \cite{HHHH09, VSK13, MES4qubit, Gour2016, sauerwein2017}.

Most of these notions have been developed considering systems of
distinguishable particles, and with system Hilbert spaces that have a
natural tensor-product structure imposed by the spatial separation of
subsystems. When applying them to systems of \emph{indistinguishable}
particles, central notions of entanglement theory have to be adapted
to account for (anti)commutation relations and superselection rules,
that restrict the set of allowed operations and modify the structure
of ``local'' operations.

In the present article, we investigate the entanglement properties of
multipartite \emph{fermionic} states. There are both fundamental and
practical reasons to do so: on the one hand, fermions are the
fundamental constituents of matter, hence to understand the
entanglement properties of quantum many-body systems the fermionic
perspective is indispensable. This has motivated a broad effort to
study fermionic entanglement and work out the differences with qubit
systems, see, e.g., \cite{ESBL02, BCW07, CamposVenuti2007, Manzano2009, Nico14, BFT14, SaLe14, DAriano2014a, EiZi15, LoFranco2016}.\\
Even in quantum information, where bosonic or effectively
distinguishable particles play the major important role, genuinely
fermionic systems such as single semiconductor electrons or holes in
quantum dots \cite{KlLo12},
ballistic electrons in quantum wires or edge
channels \cite{HT+Meunier11,MK+Ritchie11,BP+Feve12} or Majorana fermions in quantum wires
\cite{Sarma2015} are of increasing interest.  On the other hand, this
analysis gives new insights into the nature of entanglement in
general and the comparison of fermionic, bosonic, and distinguishable
systems affords a clearer picture of the role of statistics.

Here we apply this state-conversion-based entanglement classification
to multipartite \emph{Gaussian} fermionic states. This important
family of states contains the eigen- and thermal states of quadratic
Hamiltonians, i.e., those describing quasi-free single-particle
dynamics. Despite their simplicity, these states comprise a large
range of different kinds of entangled states, including GHZ-like states,
spin-squeezed states, paired states \cite{KWCG08}, and topological states \cite{WTSC13}, thus
serving as a convenient test-bed for entanglement studies and can be used
for basic quantum information processing tasks such as
probabilistic teleportation
\cite{MRZ08}, entanglement distillation \cite{KKS12},
or metrology
\cite{Yur86,KWCG08,BFM14}, 
while for universal quantum computation, the Gaussian states and
operations have to be augmented by a non-Gaussian measurement \cite{DiTe05}. In the
present work, we focus first on pure $n$-partite states with a single mode
per party, and investigate their transformation properties under
different kinds of local fermionic operations. Then we generalize some of the results to $m$-mode $n$-partite (for $m>n$) states.

The outline of the remainder of the paper is the following. In
Sec.~\ref{sec:prelim} we recall the definition and some properties of
fermionic states (FS), Gaussian fermionic states (GFS) and Gaussian
operations. Moreover, we recall the mapping between GFS and spin
states using the Jordan-Wigner transformation. In
Sec.~\ref{sec:Mixed}, we consider mixed GFS and first recall the
various definitions of separability for FS \cite{BCW07}. We identify a reasonable definition of separability of GFS. Then, we
introduce a standard form for the CM, which is invariant under GLU. In
the last two sections, Sec.~\ref{sec:purestates} and
Sec.~\ref{Sec:multimode}, we investigate the entanglement properties
of pure GFS considering GLOCC. As this class of operations turns out
to be trivial for $n$-mode $n$-partite as well as multipartite
multimode pure fully entangled GFS we study also GSLOCC and FLOCC to still obtain
insights into the entanglement of GFS. In particular, the following results are presented: (i) We characterize
  the separable Gaussian fermionic states and different kinds of local
  Gaussian fermionic operations (GLOCC, GSLOCC, GSEP); (ii) we derive a
  standard form for $n$-mode, $n$-partite GFS into which any such state can be
  transformed by GLU; as this standard form is unique, two GFS are GLU--equivalent iff their standard forms coincide; (iii) we show that there are no non-trivial Gaussian fermionic LOCC
  transformations between fully entangled pure $n$-partite GFS; (iv) we characterize the
  Gaussian SLOCC classes for pure $n$-mode, $n$-partite GFS; (v)
  we consider general fermionic LOCC between Gaussian states and identify the
corresponding maximally entangled set (MES), and, finally, (vi) we generalize these findings to the $m$-mode $n$-partite ($m > n$) case.


\section{Preliminaries}
\label{sec:prelim}
We summarize here some results concerning GFS and introduce our notation. We consider systems composed of $n$ fermionic modes.
To each mode $k=1,\dots,n$ belongs a creation and an annihilation operator $b_k,b_k^\dagger$, obeying the anticommutation relations $\{b_k^\dagger, b_l^\dagger\}=\{b_k, b_l\}=0, \{b_k, b_l^\dagger\}=\delta_{kl}$. The antisymmetric Fock space over $n$ modes is spanned by the Fock basis defined as
\bea \ket{k_1,\ldots,k_n}=(b_1^\dagger)^{k_1}\cdots (b_n^\dagger)^{k_n}\ket{0},\eea
where $k_i\in \{0,1\}$ for all $i\in \{1,\ldots,n\}$ and the vacuum state $\ket{0}$ obeys $b_i \ket{0}=0$ $\forall i$. Note that $\ket{k_1,\ldots,k_n}$ is an eigenstate of all number operators $n_i=b_i^\dagger b_i$ to eigenvalue $k_i$.

It is sometimes more convenient to consider the $2n$ hermitian fermionic
Majorana operators, \bea \tilde{c}_{2k-1}=b_k+b_k^\dagger,\,\,\,\,
\tilde{c}_{2k}=-i(b_k-b_k^\dagger)\eea instead of the creation and annihilation
operators.  The anticommutation relations are then equivalent to \bea
\{\tilde{c}_k,\tilde{c}_l\}=2 \delta_{kl}.\eea
For any Clifford algebra satisfying the relation above, the operators
$b_k=\frac{1}{2}(\tilde{c}_{2k-1}+i\tilde{c}_{2k})$ obey the anticommutation relation and vice
versa.

A linear transformation of the fermionic operators $\{\tilde{c}_k\}$,
i.e., $\tilde{c}_k\rightarrow \tilde{c}_k^\prime =\sum_l O_{kl} \tilde{c}_l$, preserves the canonical
anticommutation relations iff $O\in\O(2n, \R)$, i.e., iff $O$ is a real
orthogonal matrix. These are called canonical transformations or
  Bogoliubov transformations. They realize a basis change in the
  fermionic phase space and can be implemented by Gaussian operations
  (see below).

\subsection{Gaussian States}

A GFS of $n$ modes is defined as the thermal
(Gibbs) state of a quadratic Hamiltonian, $H=\frac{i}{4} \tilde{c}^T G \tilde{c}$ with
$G$ a real antisymmetric $2n\times 2n$ matrix and
$\tilde{c}=(\tilde{c}_1,\ldots,\tilde{c}_{2n})$, i.e., \bea \rho=K e^{-\frac{i}{4} \tilde{c}^T G
  \tilde{c}},\eea where $K$ denotes a normalization constant (or, to include
states of non-maximal rank, can be expressed as a limit of such
expressions). Equivalently, they can be characterized as those states
satisfying Wick's theorem, i.e., for which all cumulants vanish
\cite{BaHe85,Rob65}.

It is well known that any real antisymmetric $2n\times 2n$ matrix can be transformed into a
normal form via a real special orthogonal matrix \cite{BoRe04b}. More precisely, there exists
a matrix $O \in \SO(2n,\R)$ such that
\bea O G O^{T}\!=\!\oplus_{k=1}^n \beta_k J_2, \mbox{where } J_2\!=\! \left( \begin{array}{cc} 0 & 1, \\
    -1 &0
  \end{array} \right),\, \beta_k \in \R.\eea
Hence, a GFS is a state of the form  \bea\label{eq:GaussDiag}
\rho=\tilde{\otimes}_{k=1}^n \rho_k^\prime,\eea where 
$\rho_k^\prime=\frac{1}{2} \left(\one- \mu_k
  [{b^\prime}_k^\dagger,b^\prime_k]\right)$ for
$\mu_k=\tanh(\beta_k/2)$.  Note that here and in the following
$\tilde{\otimes}$ denotes a product of operators which are acting only
on distinct sets of modes. However, we only use this notation if the
operators fulfill a commutation relation. Here, the operators
$b^\prime_k=\sum_l u_{lk} b_l$ obey again the anticommutation
relations, i.e., they are fermionic annihilation operators
\cite{BoRe04b}. 
Thus, $\rho$ can be written as \bea \rho=\frac{1}{N} e^{-\sum_k \beta_k
  {b^\prime}_k^\dagger b^\prime_k},\eea where $N=\prod_k(1+e^{-\beta_k})$
denotes a normalization constant. It is evident that a Gaussian state is
completely determined by its second moments, which are usually collected in
the covariance matrix (CM). In terms of the Majorana operators the CM
of a GFS, $\rho$, which we denote by $\gamma$, is defined as 
\bea
\gamma_{kl}=-\frac{1}{2i} \tr(\rho [\tilde{c}_{k},\tilde{c}_l]).\eea
As can be easily seen from this definition, the CM is an antisymmetric
$2n \times 2n$ real matrix, which can be transformed by a real special
orthogonal matrix, $O$, into the normal form 
\bea\label{eq:Williamsonform} O\gamma O^{T}=\oplus_{k=1}^n (-\mu_k J_2).\eea
Note that $\mu_k=\tanh(\beta_k/2)$ \cite{KWCG08}. This normal form is
referred to as the (fermionic) Williamson normal form \cite{BoRe04b}. Note that in
contrast to the case of bosons, no first moments have to be specified for
fermions since due to the parity superselection rule all physical states have
$\tr(\tilde{c}_k\rho)=0$. Thus,  GFS are completely characterized by
their second moments, i.e., their CM, due to Wick's theorem
\cite{BLS94}
\bea i^p \tr(\rho \tilde{c}_{j_1} \cdots \tilde{c}_{j_{2p}})=\Pf(\gamma_{j_1,\ldots,
  j_{2p}}),\eea where $1\leq j_1<\ldots <j_{2p}\leq 2n$ and
$\gamma_{j_1,\ldots, j_{2p}}$ is the $2p \times 2p$ submatrix of $\gamma$ with
rows and columns $j_1,\ldots, j_{2p}$. Here $\Pf$ denotes the Pfaffian
which for a $2n\times 2n$ matrix $A=(a_{i,j})$ is defined as
$\Pf(A)=\frac{1}{2^n n!}\sum_{\pi\in
  S_{2n}}\mathrm{sgn}(\pi)\prod_{i=1}^{n}a_{\pi(2i-1),\pi(2i)}$,
where the sum is over all permutations $\pi$ and $\mathrm{sgn}(\pi)$
the signature of $\pi$ and satisfies $\Pf(A)^2=\det(A)$.

Note that an antisymmetric real matrix $\gamma$ corresponds to the CM
of a GFS, in particular to a normalized positive semidefinite
operator, iff $\gamma^2 \geq -\one$, i.e., iff all the eigenvalues of
$\gamma$, which are all purely imaginary, have modulus smaller or
equal to one. The CM corresponds to a pure state if $\gamma^2 =
-\one$. That is $\lambda_k\in [-1,1]$ in Eq.~(\ref{eq:Williamsonform})
and $|\lambda_k|=1\forall k$ in case the state is pure. For instance,
the CM corresponding to the vacuum state would be
$\gamma =-J_2$, whereas the one corresponding to $\proj{1}$ would be
$J_2$. Hence, the CM corresponding to the completely mixed state is
$\gamma=0$.

\subsection{Jordan-Wigner Transformation}
\label{sec:JW}
Let us recall here that there exists a one--to--one mapping between FS
and qubit states, which is known as Jordan--Wigner transformation. Let
us consider $n$ modes and define the operators

\bea\label{eq:JWT1} c_{2j-1}=Z\otimes Z\otimes...\otimes Z \otimes  X_j \otimes \one...\\ \nonumber
c_{2j}=Z\otimes Z\otimes...\otimes Z \otimes Y_j \otimes \one....\eea

These operators obey the same anticommutation relations as the Majorana operators.

Consider now a FS \bea \label{eq_FermiState}
\ket{\Psi}=\sum_{i_1\ldots,i_n\in{0,1}} \alpha_{i_1\ldots,i_n}
(b_1^\dagger )^{i_1}(b_2^\dagger )^{i_2}\ldots (b_n^\dagger
)^{i_n}\ket{0},\eea where $\ket{0}$ denotes the vacuum state and
$\alpha_{i_1\ldots,i_n}\in \C$. Due to the fact that
$(b_k^{\dagger})^2=0$ one can associate to the state given in
Eq.~(\ref{eq_FermiState}) the $n$-qubit state

\bea \label{eq_FermiState2} \ket{\Psi}=\sum_{i_1\ldots,i_n} \alpha_{i_1\ldots,i_n} \ket{i_1\ldots,i_n}_{1,\ldots,n} .\eea

The Jordan-Wigner transformation is a unitary mapping between the antisymmetric Fock space
 of $n$ modes and the Hilbert space of $n$ qubits, relating Fermi operators $\tilde{c}_i$ with qubit operators in \Eqref{eq:JWT1} and the states in \Eqref{eq_FermiState} with the ones in \Eqref{eq_FermiState2}.

Note, however, that the parties are ordered and one cannot simply
reorder them, as the order is fixed due to the commutation relations.
To give an example, the state $\ket{00}_{12}+\ket{11}_{12}=
\ket{00}_{21}-\ket{11}_{21}$, where the minus sign results from
permuting particle one and two. To be more precise, the operation
which has to be performed on the qubit state in order to swap two
systems is the fermionic swap, which is the mapping
$\ket{ij}\rightarrow (-1)^{ij} \ket{ji}$.
In order to perform, for instance, a partial trace, also these commutation relations have to be taken into account. That is, first the party over which the trace is performed has to be swapped (with a fermionic swap) to the last position \footnote{Note that the party over which one performs the trace has to be mapped to the last position. This ensures that the expectation value of all operators $A$ acting only on the first part of a bipartition $A | B$ fulfills  \unexpanded{$\left< A \right> = \tr( A \rho) = \tr(A \rho_A)$}, with \unexpanded{$\rho_A = \tr_B(\rho)$}, see \cite{Nico14, ChiralPEPS15}.}. After that, the partial trace can be performed as usual.
Fermionic mixed states are then convex combinations of fermionic pure states.

Note that the parity conservation implies that a FS is always a direct sum of
states whose support is only in the subspace with even parity and states whose
support is only in the subspace with odd parity. Here, the subspace with even (odd)
parity coincides with the set of states which are a superposition of
Fock states which have all an even (odd) number of 1's,
respectively. Denoting by $P_e $ ($P_o$) the projector onto the even (odd)
subspace we hence have that a state with density matrix $\rho$ is fermionic iff $\rho=P_e \rho P_e + P_o \rho P_o$, i.e., iff $P_e \rho P_o=P_o \rho P_e=0$. \footnote{This can be easily seen by writing the parity operator, \unexpanded{$P = i^n {\prod_{k}} c_k$, as $P = P_e - P_o$} and thus, $[\rho,P] = 0$ iff $\rho=P_e \rho P_e + P_o \rho P_o$.}

Especially in case one is working with this representation it is important to
be able to identify which of the FS are Gaussian. Fortunately, given a FS, the following result can be used to decide whether it is Gaussian or not. Recall that any operator in the Clifford algebra generated by the Majorana operators $\tilde{c}_i$ ($i=1,\ldots, 2n$) can be written as
\bea
x=\alpha \one + \sum_{p=1}^{2n} \sum_{ 1\leq a_1 <a_2 < \ldots a_p\leq 2n} \alpha_{a_1,\ldots,a_{p}} \tilde{c}_{a_1}\cdots \tilde{c}_{a_p}.\eea

An operator is called even if it involves only even powers of the generators,
or stated differently and using the Jordan-Wigner transformation, if the
number of $X$'s plus the number of $Y$'s occurring in the sum is
even. As any odd operator changes the parity, it is easy to see that $x$
  is even iff $P_exP_o=P_oxP_e=0$. Thus, in particular, all FS have even density matrices.

It has been shown in \cite{Brav05} that an even operator, $x$, is Gaussian iff
\bea [\Lambda,x^{\otimes 2}]=0, \label{GaussOp}\eea where \bea \Lambda= \sum_{i=1}^{2n} c_i \otimes c_i.\eea

Thus, we have that a FS, $\rho$, is Gaussian iff \bea \label{cond_Gauss}
[\Lambda,\rho^{\otimes 2}]=0.\eea

Let us give some examples. For a single mode, a state is fermionic if its
density matrix is diagonal in the computational basis. Any such state is also Gaussian.
For two modes, any FS is of the form $\rho=\rho_e\oplus \rho_o$ where $\rho_e$ ($\rho_o$) are density operators in the two-dimensional even (odd) parity subspace spanned by $\{\ket{00},\ket{11}\}$ ($\{\ket{01},\ket{10}\}$) respectively. It can be easily seen that such a state is then Gaussian, i.e., fulfills the condition given in Eq.~(\ref{cond_Gauss}) iff $|\rho_e|=|\rho_o|$, where $|\cdot |$ denotes the
determinant.  An example of such a state would be $e^{i\alpha ( b_1^\dagger b_2^\dagger +b_1 b_2)}$, where $\rho_e= \left( \begin{smallmatrix} \cosh(\alpha)& -i \sinh(\alpha) \\ i \sinh(\alpha) & \cosh(\alpha) \end{smallmatrix} \right)$ and $\rho_o=\one$. In particular, \emph{all} pure two-mode FS are Gaussian. However, not all mixed two-mode FS are: Examples of non-Gaussian FS are the Werner states, $\rho_W = \frac{4 F -1}{3} \proj{\psi^-}+ \frac{1-F}{3}\mathbbm{1}$, for $F \in (1/4, 1)$. Moreover, as we will see later, any pure FS of three modes is Gaussian. However, this is not the case for four modes.

When discussing pure GFS we either consider the Jordan-Wigner representation of the FS, or the CM of the state.

\subsection{Gaussian operations}
\label{GaussOper}
Let us now briefly recall the definitions and properties of Gaussian unitary
operations, general Gaussian operators and Gaussian maps in the fermionic
case. First note that
  all quantum operations (completely positive maps) that respect parity are considered as valid physical
  operations here, and referred to as \emph{fermionic operations}. 

\emph{Gaussian operations} are those that can be realized with
Gaussian means: evolution under quadratic Hamiltonians, adjoining of
systems in Gaussian states, discarding of subsystems, measuring
Gaussian POVMs, and projecting
on pure Gaussian states.
A Gaussian fermionic unitary, $U$, acting on $n$ modes can be written as
$U=e^{-iH}$, where $H$ is quadratic in the Majorana operators, that is, \bea
H=i\sum_{kl} h_{kl} \tilde{c}_k \tilde{c}_l,\eea with $h$ being a real antisymmetric $2n
\times 2n$ matrix \cite{jozsa2008matchgates}.
In \cite{terhal2002classical} it was shown that these unitaries effect a canonical
transformation of the Majorana operators
\begin{equation} \label{eqGLU1}
U^\dagger \tilde{c}_jU=\sum_{k=1}^{2n}O_{jk}\tilde{c}_k,
\end{equation}
where $O=e^{4 h}\in \SO(2n)$ is a real special orthogonal $2n\times 2n$
matrix.  Hence, a fermionic Gaussian unitary maps the CM to $O \gamma O^T$,
where $O \in \SO(2n)$ \footnote{Note that the operator $O$ can be written as $O=O_1 O_p O_2$, with
  \unexpanded{$O_i\in \SO(2n, \R)\bigcap SP(2n, \R)$}, for $i=1,2$ and  \unexpanded{$O_p \in \SO(2n,
  \R)$} the so called pairing operator. The real orthogonal matrices $O_i$ are
  passive transformations, i.e., they commute with the number operator,
  whereas the pairing operator stems from, e.g.,  \unexpanded{$b_1^\dagger b_2^\dagger$}.
Note that for a single mode, i.e., $n=1$, any real symplectic matrix is a
special orthogonal matrix.}.  \\
All Gaussian unitaries preserve the parity, i.e., they commute
  with the parity operator $P=(-1)^{\sum_kn_k}$. However, the
  parity-flipping transformation of mode $k$, which corresponds to an (non-special) orthogonal transformation $O=\oplus_{i=1}^{k-1}\one\oplus Z \oplus_{i=k+1}^{n}\one$
on the 
Majorana operators of the system \footnote{The order of the Majorana
  operators is here  \unexpanded{$\tilde{c}_1 \tilde{c}_2 \tilde{c}_3
  \tilde{c}_4\ldots \tilde{c}_{2n-1} \tilde{c}_{2n}$}.}
(here and in the following $X, Y, Z$ denote the Pauli operators) also has
a (local)  physical realization. The transformation can be achieved for
example by adjoining an ancillary mode in a Fock state
and then acting on the Majorana operators of the system modes and the
ancillary mode  with the $\SO(2n+2)$ operation $O=\oplus_{i=1}^{k-1}\one\oplus
  Z\oplus_{i=k+1}^{n}\one\oplus Z$ \footnote{The corresponding unitary is proportional
    to $\tilde{c}_{2k} \tilde{c}_{2n+2}$. Note that \unexpanded{$\tilde{c}_{2n+2}$} is a Majorana operator of
    ancillary mode which has to be ranked last. Note that the operation $\tilde{c}_{2k}\tilde{c}_{2k+1}\tilde{c}_{2k+2}\ldots
      \tilde{c}_{2n-1}\tilde{c}_{2n}\tilde{c}_{2n+2}$, which is GLU to $\tilde{c}_{2k} \tilde{c}_{2n+2}$, corresponds to a $X$ on particle $k$ in the JW-representation (see Sec.\ \ref{sec:JW}).}.
This exchanges particles with holes both in mode $k$ and in the ancilla and
leaves the latter unentangled, i.e., after discarding
the ancilla it realizes $Z$ on mode $k$. 
Since for any $O\in\O(2n)$ there exists a $O'\in\SO(2n)$ such that
$O=(\oplus _{i=1}^{n-1}\one\oplus Z)O'$ we can allow for all
orthogonal operations.
That adjoining local
ancillas enlarges the set of implementable unitaries is in contrast to the Gaussian
bosonic states and also to systems consisting of qudits. 
Hence, the most general operation on a single mode can be written as $\bar{O}=Z^{m} O$, where $m\in
\{0,1\}$, i.e., an arbitrary real orthogonal matrix. Clearly these operations
no longer correspond to unitaries which are generated by quadratic
Hamiltonians on the system modes alone [see Eq.~(\ref{eqGLU1})]. However,
as they can be implemented using a quadratic Hamiltonian and ancillas
in a Gaussian state we consider them as GLUs and take them into 
account in the following. If it is, however, the case that a particle-number
superselection rule would forbid these kind of transformations, it would be
straightforward to slightly modify the results derived here to exclude any
operation which is not of the form given in Eq.~(\ref{eqGLU1}).\\
Let us also note here that the action of any Gaussian unitary in the
Jordan-Wigner representation corresponds to a product of nearest neighbor
match gates \cite{*[{see }] [{ and references
    therein}] jozsa2008matchgates},
which are unitaries of the form $U=U_e\oplus U_o$, where both, $U_e$ and $U_o$
are $2\times 2$ unitary operators acting on the even and odd
subspace, respectively; and moreover, $|U_e|=|U_o|$.\par
A general Gaussian operator is any operator of the form
$x=e^{i\sum_{i,j} \chi_{ij} \tilde{c}_i \tilde{c}_j}$ for a complex
antisymmetric matrix $\chi$. \\ 

In \cite{Brav05} the most general Gaussian maps have been characterized via
the Choi-Jamiolkowski (CJ) isomorphism. Recall that a completely
positive (CP) map is called Gaussian if it maps Gaussian states to
Gaussian states. We reconsider in
Appendix~\ref{App:A} the CJ isomorphism for GFS. It follows that a map
 ${\cal E}$ mapping $n$ to $m$ modes
is Gaussian iff the corresponding CJ state is Gaussian (see also
\cite{Brav05}), i.e., if it is given by the
CM $E_{\cal E} = \left( \begin{smallmatrix} A & B \\
    -B^T & D \end{smallmatrix} \right)$, with $A, B, D$
$2m\times 2m, 2m\times 2n$, and $2n \times 2n$ matrices,
respectively (for more details see Appendix\ \ref{App:A} and also
\cite{Brav05}). \\ 
Note that the condition for Gaussian maps to map {\it every} Gaussian state to a Gaussian is very stringent. Consider for instance the situation where one wants to transform the state $\ket{00}+\ket{11}$ into a state $ \alpha \ket{00}+\beta \ket{11}$. Note that these are 2-mode GFS in the Jordan-Wigner representation and such a transformation is always possible for 2-qubit states via LOCC. The local operations accomplishing this transformation, i.e., $ A_1=\mbox {diag }(\alpha,\beta), \ A_2=\mbox {diag }(\beta,\alpha),$ are Gaussian, however, the map \bea {\cal E}(\rho)=A_1\otimes \one (\rho)A_1^\dagger\otimes \one+ X A_2\otimes X (\rho)A_2^\dagger X \otimes X \eea is non--Gaussian even though both terms in the sum are. A simple example of a GFS that is not mapped to a GFS by ${\cal E}$ is the 2-mode GFS $\rho = \rho_e \oplus \rho_o$, with $\rho_e = \left(\begin{smallmatrix} z_e + 1/4 & 0 \\ 0 & 1/4-z_e \end{smallmatrix} \right), \ \rho_o = \left(\begin{smallmatrix} z_o + 1/4 & x_o \\ x_o & 1/4-z_o \end{smallmatrix} \right)$ for $z_e \leq 1/4, \ \sqrt{x_o^2+z_o^2} \leq 1/4$ and $z_e^2 = x_o^2+z_o^2$, $x_o \neq 0$. Note that ${\cal E}$ is FLOCC, i.e., it is a local map which maps FS to FS (as it preserves parity), and would accomplish the desired transformation.
Due to that, we consider in Sec. \ref{sec:purestates} not only GLOCC, but also the richer class of FLOCC.


\section{Separability of Gaussian Fermionic States and Operations}
\label{sec:Mixed}
Here we specialize the three definitions of separability of general FS
presented in \cite{BCW07} to the case of GFS.  We show that they do
not all coincide even for Gaussian states and that one of them is not
stable when considering multiple copies of a state. We show that one
of the two remaining definitions of separability is also consistent
with the desired property that any separable state can
be generated by a local operation. Furthermore, we derive a standard form
for mixed $n$-mode $n$-partite states into which any GFS can be
transformed via GLU.

\subsection{Mixed Gaussian Fermionic Separable States}
\label{sec:mixed-state-ferm}

The notion of entanglement is complicated for fermions (compared to bosons or
qubits) due to superselection rules and anticommutation relations. The former
enforces that all physical states have to commute with the parity operator but
prevents that all states can be uniquely characterized by local measurements
of ``physical'' observables, i.e., those commuting with the parity
operator. The latter implies that observables
  acting on different sites (disjoint sets of modes) do not, in general,
commute.

In \cite{BCW07} several notions of product state and separable state were
discussed for arbitrary FS, i.e., not necessarily Gaussian states. There, the
set of physical states was defined as $\Pi := \{ \rho: [\rho, P] =0 \}$, with
$P$ the parity operator. This gave rise to two notions of ``product
  states'': The set of physical states for which the expectation values of
all products of physical observables factorize, i.e., $\rho(A_\pi
B_\pi)=\rho(A_\pi)\rho(B_\pi)$, was denoted by $\cP1_{\pi}$. $\cP2_{\pi}$
($\cP2$) is the set of all states of the form $\rho=\rho_A\tilde{\otimes}
\rho_B$ with (without) the parity restriction, respectively.

Then the three separable sets $\cS1_{\pi}, \cS2_{\pi}, \cS2_{\pi'}$ can be
defined via the convex hull of the different product
sets together with the requirement that the final state commutes with the
global parity.  Specifically: $\cS1_{\pi}=co(\cP1_{\pi})$,
$\cS2_{\pi}=co(\cP2_{\pi})$ and $\cS2_{\pi'}=co(\cP2_{})\cap\Pi$.

Let us now investigate these definitions further by considering GFS. In order
to identify the set of separable GFS one might want to define the separable
states as those that are not useful for any quantum information task even if
arbitrarily many copies of the state are given. Another reasonable choice
would be to define the set of separable states to be those that can (at least
asymptotically) be prepared by LOCC. In the single copy case \cite{BCW07}
shows that these two notions do not coincide for fermions: $\cP2$ contains
states that cannot be prepared locally but the set $\cP1_{\pi}$ of states that
are not useful (considering only a single copy) is strictly larger. Before we
focus on the first choice, i.e., on $\cP2$, and show that the definition using $\cP1_{\pi}$ can be ruled out,
let us present some observations about these sets.

\begin{observation}
A GFS is in the set $\cS2_{\pi}$ iff its covariance
    matrix takes direct-sum form.
\end{observation}
This can be easily seen by noting that all states in $\cS2_{\pi}$ are convex combinations of products of states that
each commute with the local parity; i.e., all terms in the mixture have a CM
that is block diagonal (and all first moments vanish), hence, the CM of the
mixture is also block diagonal.
In contrast, even the states in $\cP1_{\pi}$ can have non-zero correlations between $A$ and $B$ as stated in the next observation.

\begin{observation}
A  state in $\cP1_{\pi}$ can have non-zero correlations between $A$ and $B$. However, in that case
the block of the CM containing the correlations between A and B has at most
\emph{one} non-vanishing singular value.\\
\end{observation}

For a proof of the above Observation see Appendix~\ref{app:proofObs2}.
An example of a Gaussian state, which is separable according to definition $\cS1_{\pi}$ but not according
to $\cS2_{\pi}$ is the 2-mode Gaussian state with CM
\[
\gamma_0=\left( \begin{array}{cccc} 0&0&1&0\\ 0&0&0&0\\ -1&0&0&0\\ 0&0&0&0
\end{array} \right).
\]
It describes a state in which one (non-local and paired) mode is prepared in a
pure Fock state and the other in the maximally mixed one. In
general, we could consider the first mode to be in a (finite temperature)
thermal state (e.g., being occupied with probability $p$), then
\[\gamma_p =(1-2p)\gamma_0.\]

The two modes are defined by the non-local $\SO(4)$
matrix
\[O=\left(\begin{array}{rrrr}
0 & 0 & 1& 0 \\
1 & 0 & 0 & 0 \\
0 & 1 & 0 & 0 \\
0 & 0 & 0 & 1
\end{array}\right),\]
which maps $O \gamma_0 O^T=\left[ -J_2\oplus 0_2 \right]$. The mode operators of
the transformed state are $(\tilde{c}_3,\tilde{c}_1)$ and $(\tilde{c}_2,\tilde{c}_4)$, i.e., the new
annihilation operators are given in terms of the old ones as
$b_1^\prime=\frac{1}{2}[b_2+b_2^\dag+i(b_1+b_1^\dag)]$ and
$b_2^\prime=\frac{1}{2}[b_2-b_2^\dag-i(b_1-b_1^\dag))$. It is readily checked that the
vacuum for these two modes in the original basis is
$\ket{0_{b_1^\prime}0_{b_2^\prime}}=(\ket{0_{b_1}0_{b_2}}+i\ket{1_{b_1}1_{b_2}})/\sqrt{2}$.
Therefore the mixed Gaussian state with CM
$\gamma_p$ is given by the mixture of
$\ket{0_{b_1^\prime}0_{b_2^\prime}}$
and
$\ket{0_{b_1^\prime}1_{b_2^\prime}}=(\ket{0_{b_1}1_{b_2}}+i\ket{1_{b_1}0_{b_2}})/\sqrt{2}$ (each with
probability $(1-p)/2$) and
$\ket{1_{b_1^\prime}0_{b_2^\prime}},\ket{1_{b_1^\prime}1_{b_2^\prime}}$ (each with
probability $p/2$).
Since the Fock states in the $b_1^\prime, b_2^\prime$ basis correspond up to a local phase gate to Bell states in the local basis the state can be seen as being GLU-equivalent to a Bell-diagonal state with entries
$((1-p)/2,(1-p)/2,p/2,p/2)$ in the $(\Phi_+,\Psi_+,\Phi_-,\Psi_-)$-basis.
For qubits, we would argue that for all $p$ the density matrix is separable
(the maximal overlap with a maximally entangled state is $\leq
1/2$). Formally, a Jordan-Wigner transformation maps the fermionic two-mode state
$\rho_{\gamma_p}$ to the separable (up to LU) Bell-diagonal two-qubit state described
above. Is the GFS $\rho_{\gamma_p}$ separable or
entangled? As we show below, it does not behave as a separable state, when
many copies are available and allows (at least for $p=0$) even to distill pure
singlets. Consequently, separability should be defined in a way that does not include these
states. Note that this has already been shown for FS in \cite{BCW07}. The following theorem proves that the statement also holds for the restricted set of GFS.

\begin{theorem}
The set of Gaussian states in $\cS1_{\pi}$ is not stable. That is, there exists a GFS, $\rho$ such that $\rho\in\cS1_{\pi}$ (even in $\cP1_{\pi}$), however,
  $\rho\tilde{\otimes}\rho\not\in\cS1_{\pi}$.
\end{theorem}

\begin{proof}
Given two copies of a Gaussian state, $\rho$, with CM $\Gamma_{\rho}=\left( \begin{array}{cc}
\Gamma_A&C\\-C^T&\Gamma_B\end{array} \right)$ and $\rank\, C=1$ then the full state now has a
rank-2 matrix $C$ and therefore is no longer in $\cP1_{\pi}$, since  we can find
a pair of local observables (commuting with local parity) for which the
expectation value does not factorize. That is, assuming $(\Gamma_{\rho})_{kl}\propto\rho(c_kc_l)\not=0$ and using Wick's theorem implies $\rho^{\tilde{\otimes}2}(c_kc'_kc_lc'_l)=-\rho(c_kc_l)\rho(c'_kc'_l)\not=0$,
where the primed operators refer to the second copy.  Hence, $\rho\tilde{\otimes}\rho\not\in\cS1_{\pi}$.
\end{proof}

This shows that any Gaussian state $\rho$ for which $\rho^{\tilde{\otimes}
  n}(A_nB_n)=\rho^{\tilde{\otimes} n}(A_n)\rho^{\tilde{\otimes} n}(B_n)\forall
A_n,B_n,n$ must have a CM $\Gamma_\rho=\Gamma_A\oplus\Gamma_B$.
We are going to show next that $\rho^{\tilde{\otimes} 2}$ is not only no longer in the set $\cS1_{\pi}$, but that it  can also be useful for quantum information theoretical tasks.

\begin{observation} Some states in $\cS1_{\pi}$ can be useful for quantum information processing. \end{observation}

Given two copies of a Gaussian state with CM $\gamma_0$, we can use local
Gaussian unitaries to
transform it to the form (now written in $2\times2$ block form)
\[
\left( \begin{array}{cccc} 0&0&\id&0\\ 0&0&0&0\\ -\id&0&0&0\\ 0&0&0&0
\end{array} \right),
\]
which now contains a pure, maximally entangled two-mode state in the first
modes of A and B \cite{MRZ08}.  These states can be used for teleportation (though only
probabilistically). This shows, that $\cS1_{\pi}$ is not a viable definition of separability.

Due to these observations it is clear that one relevant set of
separable states is defined via ${\cal S}2_{\pi}$. Hence, we use
this definition in the following. In that case the CM of any
$n$-partite mixture of product states has direct-sum form, i.e.,
$\gamma = \oplus_i \gamma_i$. That is a GFS is separable iff its CM is
of that form. Moreover, this definition of separability is
meaningful in the context of the generation of separable states, as
all these states can be prepared locally. 
To be more precise, let us note that as separability does not have such a clear meaning for FS, as it has, e.g., in the bosonic, or finite dimensional case, it is a priori not clear how separable maps ought to be defined. This is especially due to the fact that the set of separable maps (SEP) does not have a clear physical meaning. In contrast to that, LOCC transformations, even if restricted to certain local operations, such as (Gaussian) fermionic operations, are operationally defined. It is the set of transformations which can be implemented by local [(Gaussian) fermionic] operations assisted by classical communication.  LOCC is strictly contained in SEP and is mathematically usually much harder to characterize. However, in a situation as here, where the definition of the larger set is not clear, one is forced to deal with LOCC. Hence, we consider in Appendix ~\ref{App:A} FLOCC transformations and show that this leads to a natural choice of the definition of FSEP. Moreover, we show that any separable state (according to ${\cal S}2_{\pi}$) can be generated via a FLOCC transformation. Hence, the definition of separable states being those which are elements of ${\cal S}2_{\pi}$ meets all the necessary requirements. Note that it is, however,
not clear if for every Gaussian state in ${\cal S}2_{\pi'}$  there
exists a decomposition into physical product states, i.e., it is not
clear whether for Gaussians the sets ${\cal S}2_{\pi}$ and ${\cal
  S}2_{\pi'}$ coincide or not (in general they do not
\cite{BCW07}). However, as mentioned above and as shown in Appendix~\ref{App:A}, all states which can be reasonably prepared locally must belong to ${\cal S}2_{\pi}$.

\subsection{Gaussian Fermionic Separable Operations}
\label{sec:GFsepOp}
As recalled in Sec. \ref{GaussOper} the CJ isomorphism
provides a one-to-one mapping between quantum states and quantum
operations. Moreover, it has been shown in the finite-dimensional case that a
map is separable, i.e., it can be written as a convex combination of local operators iff the corresponding CJ state is separable
\cite{Cirac2000}. In Appendix~\ref{App:A} we show that the CJ
state of a Gaussian  separable map has a CM of the form
\bea \Gamma_{AB}=\Gamma_A \oplus \Gamma_B,\eea with a natural generalization
to more systems. As a consequence of the previous section this state is a
separable GFS according to ${\cal S}2_{\pi}$. Thus, this definition of
separability agrees with the operational viewpoint that all separable states
can be generated locally (an agreement which is not maintained for all
  definitions in the presene of superselection rules, see, e.g., \cite{SVC04}). Moreover, this definition can be naturally generalized to Gaussian separable maps (GSEP) (see Appendix~\ref{appGLOCC} for more details).
As stated in the following lemma, fermionic completely positive maps (FCPM),
i.e., CP maps that map FSs onto FSs, can be written in
Kraus decomposition with special Kraus operators (see also \cite{DAriano2014a}).
\begin{lemma} \label{le:KrausOpParity}
All fermionic completely positive maps can be written using only Kraus operators with definite
parity (i.e., that are either sums of only even monomials in the Majorana operators $\tilde{c}_i$ or sums
of only odd monomials).
\end{lemma}

\Proof: Let $\cE$ denote a FCPM with Kraus operators $\{A_k\}$, i.e.,
$\cE(\rho)=\sum_k A_k\rho A_k^\dag$ for all $\rho$. In general, the $A_k$ are sums
of even and odd terms, that is $A_k=A^{(e)}_k+A^{(o)}_k$. FCPMs map
FSs to (unnormalized) FSs, i.e., both $\rho$ and
$\cE(\rho)$ are even (sums of even monomials in the Majorana operators $\tilde{c}_i$). Using the Kraus
representation, this implies that
$\sum_k A^{(e)}_k\rho (A^{(o)}_k)^\dag+A^{(o)}_k\rho (A^{(e)}_k)^\dag = 0$ for
all $\rho$. Consequently, the FCPM $\tilde{\cE}$ with Kraus operators $\left\{
  A^{(e)}_k, A^{(o)}_k \right\}$, which we denote by $\tilde{A}_k$ in the following, represents the same channel as
$\cE(\rho)=\tilde{\cE}(\rho)$ for all FSs $\rho$. To show that $\sum_k \tilde{A}_k^{\dagger} \tilde{A}_k  = \mathbbm{1}$ on the whole state space, note that $\tr(Y \rho) = 1 \ \forall \rho = \rho_e \oplus \rho_o$ iff $Y = \left(\begin{smallmatrix} \mathbbm{1}& Y_{e o} \\ Y_{o e} & \mathbbm{1} \end{smallmatrix} \right)$, where $Y_{eo} = P_e Y P_o$ ($Y_{oe} = P_o Y P_e$) and both the even and odd part of $Y$ have to be equal to the identity. Moreover, for $Y = \sum_k \tilde{A}_k^{\dagger} \tilde{A}_k = A_e \oplus A_o$ it follows immediately that $Y_{eo} = Y_{oe} = 0$. Thus, the Kraus operators of the FCPM $\tilde{\cE}$ also satisfy $\sum_k \tilde{A}_k^{\dagger} \tilde{A}_k  = \mathbbm{1}$.
\qed

\subsection{Standard Form and GLU--Equivalence for $n$-mode $n$-partite States}
\label{sec:standardform}
Here, we consider $n$-mode $n$-partite fermionic systems. That is, each mode
is spatially separated from the others. We derive a standard form $S(\gamma)$
into which the CM $\gamma$ can be transformed via GLU. As the
standard form is unique, we have that two GFS are GLU equivalent iff their
CMs in standard form coincide.

Let us start by recalling that the most general GLU operation corresponds to
an arbitrary real orthogonal matrix on each mode. Hence, the CM $\gamma$ is transformed to
\bea S(\gamma)= \left(\oplus_i Z^{m_i} O_i\right) \gamma  \left(\oplus_i  O_i^T Z^{m_i}\right), \eea
via GLU with $O_i \in \SO(2,\R)$ and $m_i\in \{0,1\}$. We denote in the
following by $\gamma_{jk}$ the $2\times2$ matrix describing the covariances
between modes $j$ and $k$. Due to the fact that $\gamma=-\gamma^T$ we have for $i\leq n$
\bea
\gamma_{ii}= \left( \begin{array}{cc} 0 & \lambda_i\\
-\lambda_i &0
\end{array} \right)=\lambda_i J_2,
\eea
where $\lambda_i\in [-1,1]$. As $A J_2A^T= |A| J_2$, for any $2\times2$ matrix $A$, $\gamma_{ii}$ transforms to $Z^{m_i} \gamma_{ii} Z^{m_i}=(-1)^{m_i}\gamma_{ii} $. If $\lambda_i \neq 0$ we chose $m_i$ such that
$Z^{m_i} \gamma_{ii} Z^{m_i}=\lambda_i J_2$, where $\lambda_i >0$. In case $\lambda_i=0$, i.e., mode $i$ is completely mixed, we show below how the bit value $m_i$  can be uniquely defined. In order to uniquely define $O_i=e^{i\alpha_i Y}$ we proceed as follows.
Consider the first index $j$ with $i<j$ such that the off--diagonal matrix
$\gamma_{ij}$ is not vanishing. If the singular values of $\gamma_{ij}$
are non--degenerate ($d_{ij}\not=|d_{ij}'|$) we define $O_i$, $O_j \in \SO(2,\R)$ by $O_i
\gamma_{ij} O^T_j=D_{ij}=\mbox{diag}(d_{ij},d_{ij}'), d_{ij}\geq
|d_{ij}'|$ \footnote{$O_i$ diagonalizes $\gamma_{ij}\gamma_{ij}^T$ and
  $O_j$ diagonalizes $\gamma_{ij}^T\gamma_{ij}$. Note that due to the
  restriction $O_i$,  \unexpanded{$O_j \in \SO(2,\R)$} $D_{ij}$ cannot
  be chosen non-negative.}. If the singular values of $\gamma_{ij}$
are degenerate, $\gamma_{ij}$ is itself proportional to an 
orthogonal matrix. In case $|\gamma_{ij}|>0$, $\gamma_{ij}$ is proportional to
a special orthogonal matrix, $e^{i\alpha_{ij} Y}$. Then, we define $O_j
\propto O_i \gamma_{ij}$, that is we set $\alpha_j=\alpha_{ij}+\alpha_i$. In
case $|\gamma_{ij}|<0$, $\gamma_{ij}$ is proportional to a matrix $Z
e^{i\alpha_{ij} Y}$. Then, we define $O_j \propto Z O_i
\gamma_{ij}$, that is we set $\alpha_j=\alpha_{ij}-\alpha_i$. In all
cases $S(\gamma)_{ij}$ is diagonal. We proceed in the same way for $\gamma_{i
  j+1}$ (and then any subsequent $\gamma_{ik}$). If $\alpha_j$ has already
been determined in a previous step, $\alpha_k$ is determined by diagonalizing
$\gamma_{jk}^T \gamma_{jk}$. More precisely, $\alpha_k$ is chosen such that $ O_j\gamma_{jk}O_k^T=\tilde{O}_{jk}D_{jk}$ with $D_{jk}=\mbox{diag}(d_{jk},d_{jk}'), d_{jk}>
|d_{jk}'|$ \footnote{If $d_{jk}=
d_{jk}'$ then $\gamma_{jk}$ is proportional to an orthogonal matrix and
$\alpha_k$ is determined as explained before.}, $\tilde{O}_{jk}\in \SO(2,\R)$
and $(\tilde{O}_{jk})_{11}\geq 0$ (for $(\tilde{O}_{jk})_{11}= 0$ choose
$\alpha_k$ such that $(\tilde{O}_{jk})_{12}\geq 0$) \footnote{Analogously one
  determines $\alpha_j$ by diagonalizing $\gamma_{jk} \gamma_{jk}^T$ (and
  imposing the same conditions on the singular values and the orthogonal
  matrix) if only $\alpha_k$ has been already determined.}.  If $\alpha_j$ has
been expressed as a function of some other $\alpha_l$, $l<j$, which cannot be
determined by the procedure explained before then we fix $\alpha_l$ by
diagonalizing $\gamma_{j k} \gamma_{j k}^T$ and imposing that the singular
values are ordered non-increasingly \footnote{One applies an analogous
  procedure if $\alpha_k$ is expressed as a function of $\alpha_l$.}. Note
that if not both $O_j$ and $O_k$ depend on $\alpha_l$ we can choose either
$(O_j\gamma_{j k}O_k^T)_{11}> 0$  or if $(O_j\gamma_{j k}O_k^T)_{11}= 0$ we
impose that $(O_j\gamma_{j k}O_k^T)_{12}\geq 0$. In case $\gamma_{jk}$ is
proportional to an orthogonal matrix then either
  one relates $O_k$ and $\alpha_l$ using the scheme explained before or $O_k$
  has already been related to $\alpha_l$ in a previous step. In the second case either  $O_j\gamma_{j k}O_k^T$ is independent of $\alpha_l$ or one chooses $O_j\gamma_{j k}O_k^T=\mbox{diag}(|d_{jk}|,d_{jk})$.\\
 It is easy to see that in this way any $\alpha_j$
is uniquely determined unless the CM is invariant under the conjugation with
$O_j$, that is, the mode $j$ is decoupled from all other modes,
in which case we set $\alpha_j=0$. At this point all the operators which are no symmetry of the CM are determined. Those which leave the CM invariant can be chosen to be equal to the identity, e.g., if for 3-modes $\gamma_{12} = O_{12},\ \gamma_{13} = O_{13}$ and $\gamma_{23} = O_{23}$ with $O_{ij}\in \SO(2,\R)$, i.e., all of them are special orthogonal matrices and invariant under $O_1$, we choose $O_1 = \mathbbm{1}$.

It remains to consider the case where $\lambda_i=0$.  If there is no index $j$
such that $\gamma_{ij}\neq 0$ then the mode $i$ factorizes and we set
$m_i=0$. Hence, let us assume that $\gamma_{ij}\neq 0$ for some $j$. We
determine $m_i+m_j$ by requiring that $Z^{m_i}O_i \gamma_{ij} O_j
Z^{m_j}=D_{ij}$ such that $\tr(D_{ij})> \tr(Z D_{ij})$. In case $m_j$ is
determined by the condition on the transformed $\gamma_{jj}$, this determines
$m_i$. Otherwise, there exists either a $k$ such that either $\gamma_{ik}\neq
0$ or $\gamma_{jk}\neq 0$ or, the modes $i$ and $j$ factorize. In this case,
the CM is invariant under the transformation $Z^{m_i} \oplus Z^{m_j}$ and we
set $m_i=m_j=0$. Note that if selection rules forbid the application of the
operations $Z$ to the individual modes, we simply set $m_i=0$ $\forall i$ in
the derivation above.

In summary, we have shown that any GFS can be easily transformed into its standard form by applying GLU. As the standard form is unique we have the following theorem.
\begin{theorem}\label{th:GLUequiv}
  Any CM $\gamma$ can be transformed into its standard form, $S(\gamma)$,
  by Gaussian local unitaries (GLUs). Two CMs $\gamma$ and $\Gamma$ are
  GLU-equivalent if and only if $S(\gamma)=S(\Gamma)$.
\end{theorem}
As the CM determines uniquely the corresponding GFS, Theorem~\ref{th:GLUequiv} presents a criterion for GLU--equivalence of GFS.\par
Let us consider now some examples, where we explicitly compute the standard form for the CM. As mentioned above we consider here $n$-mode $n$-partite systems, i.e., the $1 \times 1\times\dots\times 1$ case. Here, we compute the standard form of 2- and 3-mode states.

\subsubsection{$1\times 1$}
Using the definition of the standard form introduced above, it is straightforward to see that any 2-mode state CM can be written (up to GLU) as
\bea
S(\gamma)= \left( \begin{array}{cccc} 0 & \lambda_1&d_{12}&0\\
-\lambda_1 &0&0&d^\prime_{12}\\
-d_{12}&0&0&\lambda_2\\
0&-d^\prime_{12}&-\lambda_2&0
\end{array} \right), 
\eea 
with $\lambda_i> 0$ for $i\in\{1,2\}$ and $d_{12}\geq |d^\prime_{12}|$ or $\lambda_i= 0$ and $\lambda_j\geq 0$ for  $\{i,j\}=\{1,2\}$ and $d_{12}\geq d^\prime_{12}\geq 0$.
Imposing that the state is pure, i.e., that $\gamma \gamma^T=\one$ we obtain $\lambda_1=\lambda_2>0$, $d_{12}=-d^\prime_{12}$ and $d_{12}^2+\lambda_1^2=1$ or $\lambda_1=\lambda_2=0$ and  $d_{12}=d^\prime_{12}=1$ (the maximally entangled state).
\subsubsection{$1 \times 1 \times 1$}
Similar to above, one can identify the standard form for mixed states of 3 modes to be
\bea\nonumber
  S(\gamma) = \hspace*{-0.1cm} \left(\small{ \begin{array}{cccccc} 0 & \lambda_1&d_{12}&0&l_1 d_{13}&l_2 d^\prime_{13}\\
-\lambda_1 &0&0&d^\prime_{12}&-l_2 d_{13}&l_1 d^\prime_{13}\\
-d_{12}&0&0&\lambda_2&m_1&m_{12}\\
0&-d^\prime_{12}&-\lambda_2&0&m_{21}&m_2\\
 -l_1 d_{13}&l_2 d_{13}&-m_1&-m_{21}&0&\lambda_3\\
-l_2 d^\prime_{13}&-l_1 d^\prime_{13}&-m_{12}&-m_2&-\lambda_3&0
\end{array} }\right).
\eea
Thus there are 13 free parameters characterizing the mixed GFS, which have to obey certain conditions, given in Appendix~\ref{3modeGLU}.\par
Imposing the condition that the state is pure is more involved than in the case of two modes.
Even though it is straightforward to derive this decomposition for the CM, we use the Jordan-Wigner representation of the states instead. In
Sec.~\ref{sec:GSLOCC} we show that any pure GFS is either of the from
$\ket{\Phi}=a_1\ket{000}+a_2\ket{011}+a_3\ket{101}+a_4\ket{110}, a_i \in \mathbb{R}_{\geq 0} \forall i, \ \sum_{i=1}^4 a_i^2=1$
or of the form $X^{\otimes 3} \ket{\Phi}$. Note that without loss of generality $1/2 (a_3^2+a_4^2- a_1^2 - a_2^2) \geq 0$, $1/2 (a_2^2+a_4^2- a_1^2 -
a_3^2) \geq 0$, $ 1/2 (a_2^2-a_4^2- a_1^2 + a_3^2) \geq 0$ 
  (equivalent to non-negative $\lambda_i$). For \emph{strict}
  inequalities and for $a_i\neq 0$ $\forall i$ (i.e., the case of a
  generic CM without degeneracies) the standard form of the CM is given by
\begin{widetext}
\begin{align} \label{eq:Sgamma}
S(\gamma)&=2  \left(\begin{array}{cccccc}0 &\!  \lambda_1 &\! a_1 a_4 \! +\!  a_2 a_3 &\! 0&\!  0&\! -a_1 a_3 \! + \! a_2 a_4\\
-\lambda_1 &\! 0&\! 0&\!  - a_1 a_4 \! +\!  a_2 a_3&\!  -( a_1 a_3\!  + \! a_2 a_4) &\! 0 \\
-( a_1 a_4 \! +\!  a_2 a_3)&\! 0&\! 0&\! \lambda_2 &\! a_3 a_4\!  -\!  a_1 a_2 &\! 0 \\
0&\!  a_1 a_4\!  - \! a_2 a_3 &\! -\lambda_2 &\! 0&\! 0 &\!  (a_3 a_4\!  +\!  a_1 a_2)\\
0 &\! a_1 a_3\!  +\!  a_2 a_4 &\! -a_3 a_4\!  +\!  a_1 a_2 &\! 0 &\! 0&\! \lambda_3\\
 a_1 a_3 \! - \!  a_2 a_4&\! 0 &\! 0 &\! -(a_3 a_4 \! +\!  a_1 a_2) &\! -\lambda_3 &\! 0
\end{array}\right),
\end{align}
\end{widetext}
with $\lambda_1 = 1/2 (a_3^2+a_4^2- a_1^2 - a_2^2), \ \lambda_2 = 1/2 (a_2^2+a_4^2- a_1^2 -
a_3^2), \  \lambda_3 = 1/2 (a_2^2-a_4^2- a_1^2 + a_3^2)$. If the above
stated conditions do not hold a similar standard form can be
derived. More precisely, if one of the $a_i$'s is equal to zero at least one of the off-diagonal blocks is degenerate and therefore, as explained above the standard form looks slightly different. Note that as in the bosonic Gaussian case \cite{GiKr14} and in contrast to the qubit case \cite{Kus} it can be easily seen that the purities of the
reduced states, that is the $\lambda_i$'s, 
 uniquely define the state. 
Let us remark here, that there exists only one GFS (up to GLU) with
$\rho_i\propto\one$ for each subsystem $i$, namely
$\ket{000}+\ket{011}+\ket{101}+\ket{110}$. Note that -- although it is
known that any three-\emph{qubit} state whose single-qubit reduced
density operators are completely mixed is LU--equivalent to the GHZ
state -- this does not immediately imply the same for GFS due to the
restriction to GLU.

\section{Pure Gaussian Fermionic States and Local Transformations for $n$-mode $n$-partite States}
\label{sec:purestates}
Let us now investigate in more detail the entanglement contained in
\emph{pure} GFS. For this purpose we consider the class of Gaussian separable
operations (GSEP). In general, SEP contains LOCC but is a strictly larger class \cite{Bennett1999, kleinmann2011, whatyouLOCC, Hebenstreit2016}. We show, however, that for Gaussian operations on $n$-mode $n$-partite
systems any transformation among pure fully entangled states via Gaussian SEP
(GSEP) can be performed via GLU. Hence, in particular, only trivial Gaussian
LOCC (GLOCC) transformations exist for single modes. Note that here and in the
following we consider only fully entangled states, i.e., states where no
subset of modes factorizes from the remainder.  Due to the triviality of GLOCC
we study then Gaussian stochastic LOCC (GSLOCC) and certain fermionic LOCC
(FLOCC', see Section \ref{sec:FLOCC}), which map FSs to FSs. We characterize the various GSLOCC classes, which are, in contrast to
the bosonic case, indeed equivalence classes \footnote{In analogy to the
  finite dimensional case and in contrast to the bosonic case all operators
  are bounded and hence invertible.}. We then show that there exist non-trivial
FLOCC transformations that map a pure GFS to some other pure GFS and demonstrate
how to identify all possible transformations of that kind. Interestingly, many
of the pure GFS belong to the Maximally Entangled Set (MES) \cite{VSK13}. That
is, they cannot be obtained from any other state via local deterministic transformations. For other
states we derive a very simple local protocol which can be used to reach the
state from a state in the MES.

Let us first of all show that Cond.~\eqref{cond_Gauss}, which is a necessary
and sufficient condition for a FS to be also Gaussian, simplifies for pure FS
(see also \cite{Terhal}).

\begin{lemma} \label{LemmaGaussianPure} Let $\ket{\Psi}$ be a FS. Then $\ket{\Psi}$ is a GFS iff
\bea \label{condGausspure} \Lambda (\ket{\Psi}\otimes \ket{\Psi})=0.\eea
\end{lemma}
\begin{proof}

As mentioned before, an even operator, $X$ is Gaussian iff
$[\Lambda,X\otimes X]=0$ (see Eq.~\eqref{GaussOp}). As the projector
onto a FS, $\ket{\Psi}$, is even and as a hermitian rank-one operator
commutes with another hermitian operator, such as $\Lambda$, iff the state in the range of the projector is an eigenstate of $\Lambda$ we have that $\ket{\Psi}$ is GFS iff
$\Lambda (\ket{\Psi} \otimes \ket{\Psi}) = a \ket{\Psi} \otimes
\ket{\Psi}$ for some $a\in \R$. As $\ket{\Psi}$ has
  well-defined parity, we have that $\bra{\Psi} c_i \ket{\Psi}=0$ for
any  operator $c_i$. Hence, $(\bra{\Psi} \Lambda \ket{\Psi})
\ket{\Psi}= 0= a \ket{\Psi}$.
\end{proof}

\subsection{Gaussian Separable Operations and Gaussian LOCC}
\label{sec:GLOCC}
Let us start with the investigation of GSEP transformations. As argued in
Appendix~\ref{appGLOCC}, GSEP is defined as the class of operations for which
the CJ state is Gaussian and has a CM of the form $\Gamma=\oplus_{i=1}^n
\Gamma_i$. We show here that any GSEP acting on $n$ separated modes, which
maps at least
one pure (fully entangled) state into a different pure (fully entangled) state
is a GLU transformation. Hence, no non--trivial state transformation is possible. The following lemma allows us to show in the end that GLOCC on pure states are trivial, as GSEP strictly includes GLOCC (see Appendix~\ref{appGLOCC}).

\begin{lemma}\label{LemmaGSEP} Let ${\cal
    E}_{sep}$ denote a Gaussian trace preserving separable map which
  transforms at least one pure $n$-partite $n$-mode FS,
  $\ket{\Psi}$, into another pure $n$--partite $n$--mode fully
  entangled FS, $\ket{\Phi}$. Then, it holds that ${\cal
    E}_{sep}(\rho)=(U_1\tilde{\otimes} U_2\dots \tilde{\otimes} U_n) \rho
  (U_1^\dagger \tilde{\otimes} U_2^\dagger\dots \tilde{\otimes} U_n^\dagger)$ for all
  $\rho$.
\end{lemma}
\begin{proof} Every separable Gaussian CP trace-preserving map
  (GCPTM) $\cE_{sep}$ has a separable Gaussian CJ state $E_{\cE_{sep}}$, i.e., $E_{\cE_{sep}}$ is of the form $\rho_1 \tilde{\otimes}\rho_2\dots \tilde{\otimes}\rho_N$,  and, consequently, $\cE_{sep}=\cE_1 \tilde{\otimes} \cE_2\dots \tilde{\otimes}\cE_N$ is a product operation with GCPTMs $\cE_i$ (see Appendix~\ref{appGLOCC}).   Let us denote $\id\tilde{\otimes}_{k\neq 1}{\cal E}_k(\kb{\Psi}{\Psi})$ by  $\rho$ and write it in its spectral decomposition $\rho=\sum_i p_i \kb{\Psi_i}{\Psi_i}$. It follows from $\cE_{sep}(\kb{\Psi}{\Psi})=\kb{\Phi}{\Phi}$ that $\sum_i p_i {\cal E}_1\tilde{\otimes}\id^{\tilde{\otimes} n-1}(\kb{\Psi_i}{\Psi_i})=\kb{\Phi}{\Phi}$. Hence, it has to hold that   for $p_i\neq 0$ ${\cal E}_1\tilde{\otimes}\id^{\tilde{\otimes} n-1}(\kb{\Psi_i}{\Psi_i})=\kb{\Phi}{\Phi}$ and therefore there exists at least one pure state $\ket{\Psi_i}$ for which \begin{eqnarray}\label{EQSEPLU}
 \ket{\Phi}\propto A_k\tilde{\otimes}\id^{\tilde{\otimes} n-1} \ket{\Psi_i}\propto A_l\tilde{\otimes}\id^{\tilde{\otimes} n-1}\ket{\Psi_i},
 \end{eqnarray}
where by $A_j$ we denote the Kraus operators of ${\cal E}_1$. Note
that $\ket{\Psi_i}$ has to be entangled in the splitting mode $1$
versus the remaining modes as $\ket{\Phi}$ is entangled in this
splitting and ${\cal E}_1$ cannot generate entanglement. Hence,
considering $\ket{\Psi_i}$ in its Jordan-Wigner representation its
Schmidt decomposition can be written as $\ket{\Psi_i}=\sum_{j=0}^1
\lambda_j^{i}\ket{j}_1\ket{\psi_j^{i}}$ with
$\lambda_0^{i},\lambda_1^{i}\not=0$. Using this in Eq.~(\ref{EQSEPLU})
as well as that due to Lemma \ref{le:KrausOpParity} the Kraus
operators of ${\cal E}_1$ can be chosen such that each of them
commutes with $\proj{\psi_j^{i}}$ (which is a sum of only even
monomials in the Majorana operators acting on the modes $2,\ldots, n$)
it is easy to see that $A_k\ket{j}\ket{\psi_j^{i}}=c
A_l\ket{j}\ket{\psi_j^{i}}$ for $j\in\{0,1\}$ and $c\in\C $
\footnote{It straightforwardly follows from this equation that the
  action of $A_k$ is the same (up to a proportionality factor) as the
  action of $A_l$ on the whole Fock basis. In order to see this, apply
  for \unexpanded{$j\in\{0,1\}$} the projector  \unexpanded{$\proj{\vec{v}_j}$} where
   \unexpanded{$\ket{\vec{v}_j}$} is a Fock state of modes  \unexpanded{$2,\ldots ,n$} for which
   \unexpanded{$\ket{\vec{v}_j}\ket{\vec{v}_j}\bra{\Psi_j^i}\neq 0$} to both sides
  of the equation. Note that   \unexpanded{$\proj{\vec{v}_j}$} commutes with the
  Kraus operators as they have definite parity. Then apply the
  operators  \unexpanded{$(\tilde{c}_{3})^{m_2}(\tilde{c}_{5})^{m_3}\ldots
  (\tilde{c}_{2n-1})^{m_n}$} for  \unexpanded{$m_i\in\{0,1\}$}  to the resulting
  equation. It is important to note here that in order for
   \unexpanded{$A_k\ket{j}\ket{\Psi_j^i}=cA_l\ket{j}\ket{\Psi_j^i}$} to possibly
  hold true either both, $A_k$ and $A_l$, are sums of even monomials
  in the Majorana operators or both are sums of odd monomials. Hence,
  whenever $A_k$ commutes (anticommutes) with
   \unexpanded{$(\tilde{c}_{3})^{m_2}(\tilde{c}_{5})^{m_3}\ldots
  (\tilde{c}_{2n-1})^{m_n}$} so does  $A_l$ respectively. This
  procedure generates the whole Fock basis on the modes $2,\ldots ,n$
  and one obtains  \unexpanded{$A_k\ket{j}\ket{\vec{w}}=cA_l\ket{j}\ket{\vec{w}}$}
  for  \unexpanded{$j\in\{0,1\}$} and an arbitrary Fock state  \unexpanded{$\ket{\vec{w}}$} on
  modes  \unexpanded{$2,\ldots ,n$}.}. As the action of the different Kraus
operators on a basis  leads to the same states (up to a constant
proportionality factor) we have that $A_k\propto A_l$. Moreover, as
this holds true for all possible pairs of Kraus operators one obtains
from $\sum_i A_i^\dagger A_i=\one$ that  $A_i^\dagger A_i\propto\one$
and hence the map ${\cal E}_1$ corresponds to the application of a GLU
on mode $1$. Rearranging of the modes such that mode $j$ corresponds
to the first mode and using the same argumentation as before shows
that ${\cal E}_j$ is a GLU transformation on mode $j$ for all
$j$. Note that here we make use of the fact that the maps ${\cal E}_j$
commute with each other \footnote{Due to Lemma \ref{le:KrausOpParity} the
  Kraus operators can be chosen with definite parity and therefore the
  operations on the different modes commute. That is,
  \unexpanded{$(\cE_1 \tilde{\otimes}\id)(\id\tilde{\otimes}\cE_2) (\cdot)=
\sum_{k,l}
(A_{1k}\tilde{\otimes}\id)(\id\tilde{\otimes}  A_{2l})(\cdot)
(\id\tilde{\otimes}  A_{2l}^\dag)(A_{1k}^\dag\otimes\id)
=\sum_{k,l}
(\id\tilde{\otimes}  A_{2l})(A_{1k}\tilde{\otimes}\id)
(\cdot)
(A_{1k}^\dag\tilde{\otimes}\id)(\id\tilde{\otimes}  A_{2l}^\dag)
=(\id\tilde{\otimes}\cE_2)(\cE_1\tilde{\otimes}\id)(\cdot)$}. This follows from the fact that we get either no or two phase
factors of $-1$ when commuting the Kraus operators. }. Hence, we can apply the $\cE_j$ sequentially in any order. This implies that under rearranging the modes the product structure of the map ${\cal E}_{sep}$ and the Kraus operators of the local maps ${\cal E}_j$ are preserved \footnote{Note, moreover, that rearranging the modes cannot transform a pure $n$--partite $n$--mode entangled GFS into a state for which one mode factorizes.}. Hence, we have that ${\cal E}_{sep}$ is a GLU transformation.
 \end{proof}

 As mentioned above, Lemma \ref{LemmaGSEP} allows us to directly obtain the
 following corollary.
\begin{corollary}There exists no non--trivial GLOCC operation mapping a pure  $n$--mode $n$--partite FS $\ket{\psi}$ into another pure $n$--mode $n$--partite FS $\ket{\phi}$.
\end{corollary}

Let us note here that a very similar result has recently been proven for finite dimensional Hilbert spaces \cite{Gour2016, sauerwein2017}. There, it has been shown that generically, i.e., for a full--measure set of states, there exists no LOCC (even SEP) transformation, which transforms one pure (fully entangled) state into another, which is not LU--equivalent. In strong contrast to the scenario considered here, the reason for that is however not that all separable maps are particularly restricted, but that generically a state has no local symmetry. The relevance of local symmetries for local state transformation is recalled in Sec.\ \ref{sec:MES}. Note, however, that in the qudit case, the result only holds generically and that there exists a zero-measure set of states which can be transformed via LOCC, whereas for FS the result holds for any state.

\subsection{Gaussian Stochastic LOCC}
\label{sec:GSLOCC}
In the previous subsection we have shown that GLOCC transformations
among pure GFS are trivial. Thus, to quantify and qualify entanglement
properties of pure GFS we have to turn to a larger class of local
operations.  To that end, we now consider Gaussian stochastic LOCC (GSLOCC) \footnote{More precisely, we consider FLOCC operations that implement probabilistically a Gaussian operator.}.

 As mentioned before, the most general Gaussian operation consists of
 attaching an auxiliary system by applying a Gaussian unitary to it
 and the system mode and measuring the auxiliary system in the Fock
 basis. Hence, the most general operations (in the $1\times 1... \times 1$ case) are in the Jordan-Wigner representation of the form
\bea D_1 X^{k_1} \otimes D_2 X^{k_2} \otimes \ldots \otimes D_n
X^{k_n}, \label{GSLOCC}\eea where $D_i$ are diagonal (with complex
coefficients as $e^{i \alpha Z}$ is a GLU) and $k_i\in
\{0,1\}$. Note that as before the $X$ operators are possible due to
the fact that the parity of the system mode can be changed with the
auxiliary system (for total parity-preserving operations we have $k_i = 1$
for an even number of $k_i$'s). Note, furthermore, that for a single mode the Gaussian operations coincide with the fermionic operations (see Sec.\ \ref{sec:FLOCC}). Given the fact that these are the most
general Gaussian local operations we have that two states can be
transformed into each other via GSLOCC if there exists an invertible
operator of the form given in Eq.\ \eqref{GSLOCC} which transforms one
state into the other (in the Jordan-Wigner representation). In particular, we have that GSLOCC is indeed an equivalence relation.

Before studying now the possible GSLOCC classes let us introduce a standard form for FS. We consider a FS in Jordan-Wigner representation. Note again that as shown in Lemma~\ref{LemmaGaussianPure} a pure FS is Gaussian iff $ \Lambda (\ket{\Psi}\otimes \ket{\Psi})=0$. Using the standard form of FS explained below together with this condition one obtains a characterization of the GSLOCC classes. We then present the different GSLOCC classes for up to four mode GFS.

The following lemma states that by consecutive application of diagonal matrices any FS can be transformed into a normal form, which can, however, also vanish. For this we need the notion of a \textit{critical} state, i.e., a state whose single system reduced states are all proportional to the identity.
\begin{lemma} \label{LemmaNormalF} Let $\ket{\Psi}$ be a fully entangled FS. Then $\ket{\Psi}$ can constructively (by applying invertible diagonal matrices) be transformed into a unique (up to LUs) critical state, $\ket{\Psi_s}$ (up to a proportionality factor $\lambda\in \C$ which can tend to 0).
\end{lemma}
\begin{proof}
The lemma follows from the normal form of multipartite states
describing finite dimensional systems presented in
\cite{FrankNormalfo}. There, it has been shown that any state can be
transformed via (a sequence of) local operations into a state whose
single system reduced state is completely mixed. In the algorithm presented in
\cite{FrankNormalfo}, which achieves this transformation, the local
determinant 1 operations are $X_i^{(k)}=|\rho_i^{(k)}|^{1/(2d_i)}
(\sqrt{\rho_i^{(k)}})^{-1}$, where $d_i$ denotes the
local dimension of system $i$ and $\rho_i^{(k)}$ the reduced state of
party $i$ in the $k$--th step of the algorithm. In order to apply this
result to FS note that the reduced state of a FS has to be
fermionic and hence diagonal. Moreover, as local diagonal operators
are fermionic operations (even Gaussian), each state during the algorithm is a FS. Hence, in each step $k$ and for each party $i$, the operators $X_i^{(k)}$ are diagonal, which proves the statement.
\end{proof}
The normal form of $\ket{\Psi}$ is given by $\lambda \ket{\Psi_s}$ (where $\lambda$ can tend to 0).

Depending on the normal form one can group states in the following three
(disjoint) classes of states: (i) stable states: These are states
belonging to a SLOCC class which contains a
critical state, which then is their normal form. Due to the Kempf--Ness theorem \cite{KempfNess}, there exists only one critical state
in a SLOCC class (up to LUs). In the following we will call this state seed state and denote it by $\ket{\Psi_s}$. That the normal form of any stable state is the corresponding seed state follows also from the Kempf--Ness theorem. The GHZ--state, $\frac{1}{\sqrt{2}}(\ket{0 0 0 0} + \ket{1111})$ is an example of a critical and therefore a stable state; (ii) semi--stable states: belong to a SLOCC class without critical
state; The normal form of these states tends to a non-zero normal form. More precisely, it tends to a seed state of a different SLOCC class (\cite{FrankNormalfo}). The 4-qubit state $\ket{\psi} = a (\ket{0 0 0 0} + \ket{1111}) + \ket{0110} + \ket{0101}$ is an example of a semi-stable state, whose normal form tends to the 4-qubit GHZ-state (see \cite{FrankNormalfo}); (iii) states in the null cone: The normal form of these states vanishes. An example of such a state is the $W$--state. \\

In the Hilbert space $\C^{d}\otimes \ldots \otimes \C^{d}$ the union of stable
states is of full measure and dense \cite{GoWa13}. Hence, for almost all
states the normal form is not vanishing. Whether the same holds true for FS is
currently not clear. Despite this, we will focus now on stable FS. However,
in the more detailed investigations of few-mode states we will also consider
semi-stable states and states in the null cone.

It follows straightforwardly from the lemma above that stable FSs can be written as $ X^{m_1}D_1 \otimes X^{m_2} D_2 \ldots \otimes X^{m_n} D_n  \ket{\Psi_s}$ with $\ket{\Psi_s}$ being critical. Note, however, that any GFS can be written as $ X^{m_1}D_1 \otimes X^{m_2}D_2 \ldots \otimes X^{m_n}D_n  \ket{\Psi_f}$ where $\ket{\Psi_f}$ is some representative (not necessarily critical) of the GSLOCC class and $m_i\in\{0,1\}$. This follows from the fact that the most general Gaussian operations are of the form $ X^{m_1}D_1 \otimes X^{m_2}D_2 \ldots \otimes X^{m_n}D_n$. The subsequent corollary allows to characterize the GSLOCC classes of stable GFS.
\begin{corollary} \label{LemmaNormalGFS} Let $\ket{\Psi}$ be a stable FS and $\ket{\Psi}= D_1 \otimes D_2 \ldots \otimes D_n  \ket{\Psi_s}$. Then, $\ket{\Psi}$ is GFS iff $\ket{\Psi_s}$ is GFS.
\end{corollary}
\begin{proof}
The ``if''--part follows from the fact that local diagonal matrices are Gaussian operations. The ``only if''--part can be seen as follows. Due to Lemma \ref{LemmaGaussianPure} we have  that $\ket{\Psi}$ is GFS iff $\Lambda (\ket{\Psi}\otimes\ket{\Psi})=0$, which is equivalent to $\Lambda (\ket{\Psi_s}\otimes\ket{\Psi_s})=0$. Hence, $\ket{\Psi}$ is a GFS iff $\ket{\Psi_s}$ is.
\end{proof}

An interesting example of a critical GFS state is the $n$-mode state $\ket{\Psi}=H^{\otimes n}\ket{GHZ}$ [with
$\ket{GHZ} = 1/\sqrt{2}(\ket{00...0}+\ket{11....1})$]. To see that $\ket{\Psi}$ is a GFS, note that
$\ket{\Psi}\propto\sum_{{\bf k}\in\{0,1\}^n} (1+(-1)^{h({\bf k})})\ket{{\bf k}}$ with $h({\bf k})$ being the Hamming weight of the bitstring ${\bf k}$. Therefore, $\ket{\Psi}$ is a FS. That
$\Lambda (\ket{\Psi}\otimes\ket{\Psi})=0$ can be easily verified by direct
computation. The
fact that the state is critical follows from the criticality of the
GHZ state. Note that the GHZ state itself is only a FS for even $n$. Moreover,
the fermionic swap applied to any two modes of $\ket{\Psi}$ (or of any critical state) is also critical. As there exists only one critical state in a
SLOCC class, this state is either LU-equivalent to $\ket{\Psi}$ or in a
different SLOCC class \footnote{A simple example is given by the $4$-mode GHZ,
i.e.,  \unexpanded{$\ket{\Psi}=H^{\otimes 4}\ket{GHZ_4}$} and the not LU-equivalent 
state that is obtained after applying a fermionic swap on the first two
modes.}. 

Let us now explicitly compute the GSLOCC classes of up to 4-mode GFS.

\subsubsection{$1\times 1$ case}
We start with the simplest case of pure 2-mode 2-partite systems.
First note that the spin representation of any FS of two modes is either of the form $\ket{\Psi_1}=\alpha \ket{00}+\beta \ket{11}$ or of the form $(\one \otimes X) \ket{\Psi_1}=\alpha \ket{01}+\beta \ket{10}$. As $\ket{\Psi_1}\propto D\otimes \one \ket{\Phi_+}$, where $\ket{\Phi^+} = 1/\sqrt{2}(\ket{00}+\ket{11}$ denotes the critical seed state of two qubits and $D=\mathrm{diag}(\alpha,\beta)$ there is only one entangled GSLOCC class.  It is easy to see that these states are all Gaussian, as $\Lambda (\ket{\Phi^+} \otimes \ket{\Phi^+})=0$ (see Lemma\ \ref{LemmaGaussianPure} and Corollary \ref{LemmaNormalGFS}).

\subsubsection{$1\times 1 \times 1$ case}

For 3-mode GFS we denote by $\ket{GHZ}_3$ the Gaussian fermionic GHZ state, i.e., $\ket{GHZ}_3=H^{\otimes 3}[1/\sqrt{2}(\ket{000}+\ket{111})]= 1/2 (\ket{000}+\ket{011}+\ket{101}+\ket{110}).$ Note that we consider from now on only even parity FS, as the odd ones are simply given by applying $X^{\otimes 3}$.
We write an arbitrary pure 3-mode (not normalized) FS as $\ket{\Psi(a_1,a_2,a_3,a_4)}=a_1\ket{000}+a_2\ket{011}+a_3\ket{101}+a_4\ket{110}$, $a_i \in \C \ \forall i$. Applying GLUs ($e^{\alpha_i Z}$) and choosing the global phase appropriately allows to chose all the parameters $a_i$ to be real and non-negative. Using Lemma \ref{LemmaGaussianPure} it can be easily seen that they are all Gaussian. Then, the following lemma characterizes all 3-mode GSLOCC classes.
\begin{lemma} There are two 3--mode entangled GSLOCC classes, the GHZ and the W class. The state $\ket{\Psi(a_1,a_2,a_3,a_4)}$ belongs to the GHZ class iff $a_i\neq 0$ $\forall i$. It belongs to the W--class iff there exists exactly one $i$ such that $a_i=0$. Moreover, the state is biseparable iff exactly two $a_i=0$ (else it is separable). \end{lemma}

\begin{proof} First consider the case where $a_i\neq 0 \ \forall i$. It can be
  easily seen that the state can be written as $D_1\otimes D_2 \otimes D_3
  \ket{GHZ}_3$ with $D_i$ invertible and hence it belongs to the GHZ
  class. Let us denote by $\ket{W}_3=1/ \sqrt{3}(\ket{011}+\ket{101}+\ket{110})$ the
  W--state. Then it is easy to see that any state
  $\ket{\Psi(a_1,a_2,a_3,a_4)}$ with exactly one $i$ such that $a_i=0$ can be
  written as $ X^{k_1}D_1\otimes X^{k_2} D_2 \otimes X^{k_3} \ket{W}_3$, where
  $k_1+k_2+k_3=0 \mod 2$ and $D_i$ diagonal and invertible. If two
  coefficients vanish the state can be written as $X^{k_1}\otimes X^{k_2}
  D\otimes X^{k_3} \ket{0}\ket{\Phi^+}$ (up to particle permutation), where
  $k_1+k_2+k_3=0 \mod 2$ and $D$ is invertible, which proves the
  statement.
\end{proof}

Note that this implies that a tripartite entangled 3--mode GFS is of the form
$D_1\otimes D_2 \otimes D_3 \ket{\Psi_f}_3$ (up to GLUs), where
$\ket{\Psi_f}_3$ is either the GHZ-- or the W--state and all $D_i$'s are
invertible. Hence, there exist, as in the qubit case, two fully entangled GSLOCC
classes. The standard forms of the corresponding CM are given in Sec.\
\ref{sec:standardform}. To give an example for the GHZ-state with $a_i= 1/2\
\forall i$ the standard form is given in Eq.~(\ref{eq:Sgamma}). A similar
standard form for the W-state ($a_1=0, a_2=a_3=a_4=1/\sqrt{3}$) can be
determined. However, it is slightly different, as in this case $\gamma_{12},
\gamma_{13}, \gamma_{23} \in \mathcal{SO}(2, \mathbb{R})$ in
Eq.~(\ref{eq:Sgamma}).

\subsubsection{$1\times 1 \times 1 \times 1$ case}

For 4 modes it is no longer true that any pure FS is a GFS. In fact, from Lemma\ \ref{LemmaGaussianPure} one easily derives the following observation.
\begin{observation} \label{Obs4modeGFS} A pure 4-mode FS, $\ket{\Psi}$ (in Jordan-Wigner representation) is Gaussian iff
\bea \bra{\Psi^\ast} (X\otimes Y \otimes X \otimes Y )\ket{\Psi}=0,\eea
where $X,Y$ denote the Pauli operators. \end{observation}
This condition, which resembles the SL-invariant polynomials \cite{FrankNormalfo} defined for qubit-states, is in fact equivalent to the condition that all reduced 3-mode states of $\ket{\Psi}$ (taking the partial trace of one party) are Gaussian.
An arbitrary 4-mode (even parity) FS is given by $\ket{\Psi} = a_1 \ket{0000} + a_2 \ket{0011} + a_3 \ket{0110} + a_4 \ket{1100} + a_5 \ket{1010} + a_6 \ket{0101} + a_7 \ket{1001} + a_8 \ket{1111}$. It can be easily seen (analogously to the 3-mode case) that any such state can be written as in the following lemma \footnote{Note that we consider again only even states, as the same holds for odd 4-mode FS.}.

\begin{lemma} \label{lemma4Modes} A pure 4-mode FS, $\ket{\Psi}$ can be written as
 \bea \label{4-modeGFS} \ket{\Psi}= X^{k_1}D_1\otimes X^{k_2}D_2 \otimes X^{k_3}D_3 \otimes X^{k_4}D_4 \ket{\Psi_f},\eea
 with $\ket{\Psi_f}$ an appropriate representative of each SLOCC class, $k_i\in\{0,1\}$ and $k_1+k_2+k_3+k_4=0 \mod 2$.
Moreover, the state is GFS iff the FS $\ket{\Psi_f}$ is.
\end{lemma}

The last conclusion follows directly from Corollary\ \ref{LemmaNormalGFS} as in the proof it has not been used that $\ket{\Psi_s}$ is critical and the local $X^{k_i} D_i$ are Gaussian operations. Note that, as in the 3-mode case, some GSLOCC classes contain a critical state, whereas others do not. Moreover, in the 4-mode case there also exist semi-stable states, i.e., states that tend to a non-vanishing normal form, even though they are not stable. Let us state the different GSLOCC classes now in more detail based on the results on 4-qubit SLOCC classes in \cite{FrankSLOCC4}.
\begin{itemize}
\item GSLOCC classes containing a critical state: \\
These are states from the SLOCC classes $G_{abcd}$ (\cite{FrankSLOCC4}), with representatives \bea \label{4qubitseed} \hspace*{-0.4cm} \ket{\Psi_f}=a\ket{\Phi^+}^{\otimes 2}\!+\!b\ket{\Phi^-}^{\otimes 2}\!+\!c\ket{\Psi^+}^{\otimes 2}\!+\!d\ket{\Psi^-}^{\otimes 2}.\eea
Note that the states $\ket{\Psi_f}$ are critical. Due to Observation \ref{Obs4modeGFS} we can easily see that the FS in Eq.~\eqref{4qubitseed} are Gaussian iff $a b + c d = 0$. Hence, either two or three of the parameters of $\ket{\Psi_f}$ can vanish, according to this necessary and sufficient condition. Whereas the states where two of the four parameters are equal to zero are still 4-partite entangled, states with three parameters being equal to zero are biseparable states.
\item GSLOCC classes containing semi-stable states: \\
As mentioned above there exist classes that contain semi-stable states (see \cite{Sawicki13} for results on semi-stable 4-qubit states). The SLOCC classes containing 4-mode entangled GFS are $L_{a b c_2}$ and $L_{a_2 b_2}$ (see \cite{FrankSLOCC4}) with representatives
 \bea \label{4qubitsemistable} \hspace*{-0.4cm} \ket{\Psi_f(a b c_2)}&\!=\!&\frac{a\!+\!b}{2} (\ket{0000} \!+\!  \ket{1111})\!+\!\frac{a\!-\!b}{2} (\ket{0011} \!+\! \ket{1100}) \! +\nonumber \\ & &c (\ket{0101} + \ket{1010})+ \ket{0110}, \nonumber \\
 \ket{\Psi_f(a_2 b_2)}&=&a (\ket{0000} + \ket{1111})+b (\ket{0101} +\ket{1010}) \nonumber \\ && +\ket{0110}+\ket{0011}.
 \eea
Note that neither $\ket{\Psi_f(a b c_2)}$ nor  $\ket{\Psi_f(a_2 b_2)}$ are critical. Using Lemma\ \ref{lemma4Modes} and Observation\ \ref{Obs4modeGFS} we find that the FS are also Gaussian iff either $a b = -c^2$ for states in $L_{a b c_2}$ or $ a^2 + b^2 =0 $ for states in $L_{a_2 b_2}$. Note that if all of the parameters of a state in $L_{a b c_2}$ ($L_{a_2 b_2}$) are equal to zero, the state is a product state (biseparable state) respectively.
\item GSLOCC classes containing states in the null cone: \\
The states in the null cone are the ones for which the normal form vanishes. For 4-mode GFS there exists, as in the 3-mode case, exactly one GSLOCC class containing these states, which is the class $L_{a b_3}$ of \cite{FrankSLOCC4} with $a= b=0$. The representative is of the form
\bea \label{4qubitnullcone} \hspace*{-0.4cm} \ket{\Psi_f}=\ket{1100}+\ket{1111}+\ket{1010}+\ket{0110}.\eea
This state is Gaussian and GLU-equivalent to the 4-qubit W-state.
\end{itemize}

Hence, for 4-mode GFS there exist infinitely many entangled GSLOCC
classes. More precisely, there are infinitely many GSLOCC classes that contain
a critical state, that is the states in these classes can be transformed into
the normal form. Furthermore, there exist infinitely many GSLOCC classes of
semi-stable states, which tend to a non-zero normal form without being
stable. There exists also a single GSLOCC class containing states in the null
cone for which the normal form vanishes.

There are less Gaussian GSLOCC classes for 4-mode GFS than there are for FS,
which is not surprising as not all FS are GFS, due to the condition given in
Eq.~\eqref{condGausspure} on $\ket{\Psi_f}$. This also implies that there exist less GSLOCC classes than SLOCC classes in the qubit case (see \cite{FrankSLOCC4}). However, as
mentioned above, there are still infinitely many such classes.  Examples of
SLOCC classes that contain FS but no GFS are those denoted by
$L_{a_20_{3\oplus 1}}$ in \cite{FrankSLOCC4} for $a\not=0$ \footnote{The class with
$a=0$ does contain GFS and in this case the states are biseparable.}.

\subsection{Fermionic LOCC operations}
\label{sec:FLOCC}
As for transformations of pure $n$-mode $n$-partite GFS there exist no
non--trivial GLOCC transformations, we consider here a larger class of
deterministic transformations and study fermionic LOCC (FLOCC)
transformations.  For such transformations the local maps that are
applied have to be fermionic and the measurement operators that are
implemented in each round have to be parity-respecting and local,
i.e., they have to be of the form $X^k D$ (in Jordan-Wigner
representation), where $k\in\{0,1\}$ and $D$ denotes here and in the
following a diagonal matrix \footnote{Note that if we consider more
  modes per site, the operations which can be applied would be of the
  form  \unexpanded{$X^{k_1}\otimes X^{k_2} T$}, where $T$ has to commute with
   \unexpanded{$Z\otimes Z$}. }. More precisely, in each round of an FLOCC
transformation one party implements locally a fermionic POVM
measurement with measurement operators that are of the form $X^k D$,
possibly discards some classical information about the outcome, then
communicates the relevant information to the other parties. These
apply depending on the measurement outcome an arbitrary local
completely positive trace-preserving (CPT) fermionic map. Note that
the Kraus operators of such maps can be chosen to be of the form $X^k
D$ (cf. Lemma \ref{le:KrausOpParity}). Note further that the
operations that are implemented in a subsequent round might depend on
the information about the prior outcomes.\\
For a concatenation of finitely many of such rounds the Kraus operators of the map that is implemented in each branch of the protocol, i.e., for a specific sequence of outcomes (taking into account that some information might have been discarded), are of the form $X^{k_1} D_1\otimes X^{k_2} D_2 \otimes \ldots  X^{k_n} D_n$. This can be easily seen as a finite product of operators of this form results in an operator of the same form. \\
In order to provide a rigorous definition of FLOCC protocols which can also involve infinitely many rounds (in analogy to the one given in
\cite{whatyouLOCC} for LOCC protocols) let us use the description of a
protocol in terms of a quantum instrument, i.e., by the family of
CP maps $\{\cE_1,\ldots, \cE_m\}$. Here, $\cE_i$ is
the CP map that is implemented in a specific branch
of the protocol denoted by $i$ and it holds that  $\sum_{i=1}^m
\cE_i$ is a trace-preserving map. Moreover, a  quantum instrument
$\mathcal{P}$ will be called FLOCC-linked to an instrument
$\tilde{\mathcal{P}}$ if $\mathcal{P}$ can be implemented by first
implementing $\tilde{\mathcal{P}}$ followed by exactly one more round
of an FLOCC protocol as defined before  (where again the operations
that are implemented in each branch $i$ can depend on all previous outcomes)  and then
possibly by some discarding of classical information.  With all that, $\mathcal{F}$ is defined as the instrument of a  FLOCC transformation if there exists a sequence of instruments of finite-round FLOCC protocols where each element of the sequence is FLOCC-linked to its preceding element. Furthermore, for each element  there exists a way to discard information in the final round such that the resulting sequence of instruments converges to $\mathcal{F}$. In the following we consider also infinite-round FLOCC, however, only those for which all Kraus operators are of the form $X^{k_i} D_i$. In order to highlight that there might be a difference to FLOCC as defined above, we denote this set of operations by FLOCC'. Note that, of course, any finitely-many-rounds FLOCC is contained in FLOCC'.
We are interested in FLOCC' transformations among pure GFS and, in particular, in the maximally entangled set for this scenario. We first review the concept of the maximally entangled set and then explain how it can be determined for GFS when one considers FLOCC' transformations.
\subsubsection{The maximally entangled set}
\label{sec:MES}
In \cite{VSK13} some of us introduced the Maximally Entangled Set (MES)
as the minimal set of $n$-partite entangled states that has the
property that any pure $n$-partite entangled state can be obtained via
LOCC from a state within this set. That is the states in the MES are
those which cannot be reached via LOCC from some state that is not
LU-equivalent.  In \cite{GiKr14} GLOCC transformations among Gaussian
states of two or three bosonic modes have been considered. There, it has
been shown that not all pure bosonic three-mode Gaussian states can be obtained via GLOCC from a symmetric Gaussian state, i.e., the MES of bosonic three-mode Gaussian states under GLOCC cannot consist only of symmetric Gaussian states.  In the following, we are interested in the MES of GFS under FLOCC'. It is defined analogously to before as the minimal set of $n$-partite $n$-mode GFS for which it holds that any pure $n$-partite $n$-mode entangled GFS can be obtained via FLOCC' from a state within this set.

As we explain in the next section using the Jordan-Wigner representation,
FLOCC' reachability of GFS can be studied in a way analogous to qubit systems. There, we used the necessary and sufficient conditions of convertibility via separable maps (SEP) of \cite{Gour} to identify the states that cannot be reached via SEP from a state that is not LU-equivalent. As separable maps (strictly) include LOCC transformations it follows that these states are not reachable via LOCC.
We outline here the basic idea of the proof of the necessary and sufficient condition derived in \cite{Gour} for qudits in order to explain how this result can also be applied to study FLOCC' transformations of GFS. \\
The initial state of the transformation is denoted by $g\ket{\Psi_s}$ and the final state by $h\ket{\Psi_s}$, where $g, h$ are invertible local operators \footnote{Note that only transformations among states within the same SLOCC class are considered here. This is necessary, as we do not consider transformations reducing the local rank of the states, i.e., we consider only transformations among truly multipartite entangled states.}. In order to perform this transformation it has to hold for all the Kraus operators of the separable map, $A_i=A_i^{(1)}\otimes A_i^{(2)}\otimes\ldots\otimes A_i^{(n)}$, that
$A_i g\ket{\Psi_s}\propto h\ket{\Psi_s}$ and therefore $(h^{-1}A_i g)\ket{\Psi_s}\propto \ket{\Psi_s}$. Using the definition for the local symmetries of a state $S_{\ket{\Psi}}=\{S:S\ket{\Psi}=\ket{\Psi}, S=S^{(1)}\otimes S^{(2)}\otimes\ldots\otimes S^{(n)}, S^{(j)}\in GL(d_j,\C)\}$, where $d_j$ denotes the local dimension of system $j$, we have that $h^{-1}A_i g\propto S_i$ where $S_i\in S_{\ket{\Psi_s}}$. That is, the measurement operators $A_i$ are proportional to $hS_i g^{-1}$. Taking into account the proper proportionality factors and using that the separable map has to be trace-preserving one obtains the following necessary condition for transforming $g\ket{\Psi_s}$ into $h\ket{\Psi_s}$ via SEP. There has to exist a probability distribution $\{p_i\}_{i=1}^m$ and local symmetries $S_i\in S_{\Psi_s}$ such that \cite{Gour} \bea\label{EqGour} \sum_{i=1}^m p_i S_i^{\dagger} H S_i=r G,\eea where $H=h^{\dagger} h$, $G=g^{\dagger} g$ and $r=\frac{\bra{\Psi_s}H\ket{\Psi_s}}{\bra{\Psi_s}G\ket{\Psi_s}}$. Moreover, it is straightforward to see that this condition is also sufficient  \cite{Gour}.  Using this criterion one can determine the states that are not reachable via a SEP transformation and hence, not via LOCC.\\
In the subsequent subsection we discuss how one can in an analogous way obtain necessary and sufficient condition for transformations among pure GFS via CPT maps with local fermionic Kraus operators.

\subsubsection{The maximally entangled set of GFS under FLOCC'}
As mentioned before the MES of GFS under FLOCC' corresponds to the minimal set of $n$-partite $n$-mode GFS with the property that any pure $n$-partite $n$-mode entangled GFS can be obtained via FLOCC' from a state within this set. Hence, this set corresponds to the optimal resource under the restriction to pure GFS and FLOCC' transformations. As we will see, it can be determined using a similar method as has been employed to characterize the MES for $3$- and $4$-qubit states.
In particular, using the Jordan-Wigner representation one can find analogously to the qudit case \cite{Gour}, which we reviewed in the previous subsection, the necessary and sufficient condition for transformations among GFS via separable maps whose Kraus operators are of the form $X^{m_1} D_1\otimes X^{m_2} D_2 \otimes \ldots  \otimes X^{m_n} D_n$, with $m_i\in\{0,1\}$ and $D_i$ is diagonal. Note that this class of separable maps includes all FLOCC' transformations, as all local fermionic operators can be written like that (in Jordan-Wigner representation).  \\
Before proceeding studying the separable maps, let us briefly recall the relation between the operator $(X)^{m_i}D_i$ in Jordan-Wigner representation and the Majorana operators. $(X)^{m_i}D_i$ corresponds to a sum of monomials of even ($m_i=0$) or odd ($m_i=1$) powers in the Majorana operators and hence, it either commutes or anticommutes with the application of $(X)^{m_j}D_j$ for $j\neq i$.  Note that as $X_i$ (in Jordan-Wigner representation) corresponds in the Majorana operators to $(-i\tilde{c}_1\tilde{c}_2)(-i\tilde{c}_3\tilde{c}_4)\ldots(-i\tilde{c}_{2i-3}\tilde{c}_{2i-2})\tilde{c}_{2i-1}$ it follows that despite the fact that this operator is acting locally on the modes it is not only acting on mode $i$. Its implementation requires also other parties to apply a local unitary. Any diagonal matrix $D_i$ can be written in the Majorana operators (up to a proportionality factor) as $e^{-i\alpha \tilde{c}_{2i-1}\tilde{c}_{2i}}$ for some $\alpha \in \C$ and therefore only acts on mode $i$.

In the previous subsection we have seen that all Kraus operators $A_i$ of a separable map transforming $g\ket{\Psi_s}$ to $h\ket{\Psi_s}$ have to be proportional to $hS_i g^{-1}$. Recall that $S_i$ denotes a local symmetry of $\ket{\Psi_s}$.
As for the transformations we are interested in  the operators $h, g$ and the Kraus operators $A_i$ are local fermionic operators this implies that also any symmetry $S_i\propto h^{-1}A_i g$ that contributes to the transformation is of the form $(X)^{m_1}D_1\otimes (X)^{m_2}D_2\ldots \otimes(X)^{m_n}D_n$. Hence, only symmetries of this form appear in the necessary and sufficient condition given by Eq.~(\ref{EqGour}) \footnote{As the SLOCC operators and local symmetries appearing in this equation are of the form   \unexpanded{$(X)^{m_1}D_1\otimes (X)^{m_2}D_2\ldots \otimes(X)^{m_n}D_n$} it can be easily seen that if one considers this equation  in terms of Majorana operators  it only involves even powers of the Majorana operators of a single mode. Hence, the local operators commute and partial traces can be performed without any additional reordering. It can also easily be seen that any operator $X$ acting on system $i$ that appears in this equation can be represented there in terms of Majorana operators by   \unexpanded{$\tilde{c}_{2i-1}$}.} if one considers transformations among GFS via the considered class of separable maps. \\
Thus, the local symmetries that can contribute to such transformations are a subset of the local symmetries that are available for transformations among qubit states. It follows straightforwardly that if the qubit state corresponding to the GFS (in Jordan-Wigner representation) is not reachable via a non-trivial SEP transformation then the GFS is not reachable via a separable map with the specific form of Kraus operators that we impose. Moreover, as exactly the same methods can be applied that we used to determine the MES for three- and four-qubits one can infer from these results the MES for $3$- and $4$- mode GFS under FLOCC' \footnote{Note that this does not imply that in order to compute the reachable GFS one simply computes the intersection of the GFS with the reachable qubit states as the required transformations to reach this state might be non-fermionic.}.\\
In \cite{finiteLOCC} and \cite{finiteroundLOCC} finite round LOCC transformations among pure $n$-qudit states have been investigated. Restricting the measurement operators, local unitaries and SLOCC operators to local fermionic operators one can use an analogous argumentation to obtain the corresponding results for finite-round FLOCC transformations among GFS.
In the following subsections we discuss explicitly the MES for $3$- and $4$-mode GFS under FLOCC'.

\subsubsection{$1\times 1 \times 1 $ case}

As shown in \cite{VSK13} the MES of three-qubit states is given (up to LUs) by
\bea\label{3MES}\hspace*{-0.4cm} \{ D_1\!\otimes \! D_2 \! \otimes \! D_3 \ket{GHZ}_3, \ket{GHZ}_3,
D_1\! \otimes \!D_2 \!\otimes \!\one \ket{W} \},\eea
where for the GHZ-class none of the $D_i$'s
is proportional to the identity and all of them are real and invertible. Note that all these states are Gaussian and it follows directly that these states also have to be in the MES of 3-mode GFS. As any GFS in the W-class can be written (up to GLUs) as given in Eq.~(\ref{3MES}), we have that any tripartite entangled 3--mode GFS is
either in the MES or it is of the form $D_1\otimes D_2 \otimes \one
\ket{GHZ}_3$, where at least one $D_i$ is not proportional to the identity (up
to GLUs and particle permutations). In the first case, the state cannot be
reached from any other state (even if one would allow the most general LOCC
transformation). In the second case it can be easily reached from the GHZ
state with the following FLOCC' protocol. Party $1$ applies the measurement
consisting of the measurement operators $D_1,  D_1 X$  and party 2 applies a
measurement consisting of the measurement operators $D_2,  D_2 X$
\footnote{More precisely, the measurement operators (including the proper
  normalization factors) for the measurement of party $i$ are given by
  $\frac{1}{\sqrt{\tr (D_i)}}D_i$ and $\frac{1}{\sqrt{\tr (D_i)}}D_i
      X$.}. Hence, the resulting state is $D_1 X^{k_1}\otimes D_2 X^{k_2}\otimes \one \ket{GHZ}_3$. Using that $ X^{k_1}\otimes  X^{k_2}\otimes X^{k_1+k_2} \ket{GHZ}_3= \ket{GHZ}_3$, we have that if party 3 applies the GLU $X^{k_1+k_2}$ the resulting state is for any outcome the desired state and hence, the transformation is deterministic.

\subsubsection{$1\times 1 \times 1 \times 1$ case}
The 4-mode case is very similar to the previously discussed 3-mode case. In
order to illustrate this, let us consider a few examples of possible
transformations among 4-mode GFS via FLOCC'. Note that we consider here only
GSLOCC classes with non-degenerate and non-cyclic seed states as in
Eq.~\eqref{4qubitseed} \footnote{That is the seed parameters fulfill $ab+cd=0$, $b^2
\neq c^2 \neq d^2 \neq b^2$, $a^2 \neq b^2, c^2, d^2$ and \unexpanded{$\nexists q \in
\mathcal{C} / \{1\}$}, such that $\{a^2,b^2,c^2,d^2\} = \{q a^2, qb^2,
qc^2, qd^2 \}$, see \cite{MES4qubit}. These conditions stem from a condition on the local symmetries of the states.}.
Due to Lemma~\ref{lemma4Modes} any 4-mode GFS with a seed state of the above form is either a
state in the MES (see \cite{VSK13}) or of the form (up to permutations) $\ket{\Psi}=D_1\otimes \one^{\otimes 3} \ket{\Psi_s}$. If the state is in the MES, it cannot be reached by any other state (even if LOCC would be allowed). Moreover, apart from the seed states all other states in the MES are isolated, i.e., they cannot be transformed into any other state via FLOCC'. Note that this is in contrast to the qubit case, where the states in Eq.~(\ref{4-modeGFS}) are states in the MES that are non-isolated, i.e., they can be transformed into a state with exactly one local non-diagonal operator (see \cite{VSK13}) via LOCC. These states are, however, no GFS. In case the 4-mode GFS is not in the MES it can be easily reached from the GFS seed state via the following FLOCC' protocol (for more sophisticated protocols see below). Party $1$ applies the measurement consisting of the measurement operators $D_1, D_1 X$. In case of the first outcome, the other parties do not need to apply any transformation. In case of the second outcome all three apply $X$ to their systems. Due to the fact that the seed state is invariant under $X^{\otimes 4}$ it can be easily seen that the transformation can be achieved deterministically.

Note that for certain GSLOCC classes more transformations are possible (see \cite{MES4qubit}). For instance, if the seed parameters fulfill $a =b, \ c= d$ and $c = i a$, that is they do not fulfill the above stated conditions, the seed state has more symmetries. As can be easily seen, this implies that the seed state can be, for example, transformed into states of the form $\one \otimes D_2 \otimes D_3 \otimes \one \ket{\Psi_s}$. The corresponding FLOCC' protocol is given by party 2 applying the measurement operators $D_2, D_2 X$ and party 3 applies the operators $D_3, D_3 X$. Using that the seed state is invariant under $Y \otimes \one \otimes X \otimes Z$ and $Z \otimes  X \otimes  \one\otimes Y$ it is easy to see that the protocol can be implemented deterministically.

\section{Pure Gaussian fermionic states and local transformations for multimode states}
\label{Sec:multimode}
In this section we consider pure $N$-partite GFS where each party $i$ holds
$m_i$ modes. We first investigate transformations among fully entangled
multimode GFS (for the definition see below) via Gaussian trace preserving separable transformations (GSEP),
i.e., Gaussian transformations for which the CM of the
CJ state is of direct sum form. We show that also in this
more general setting such transformations are only possible if the map is a
GLU transformation. As GSEP includes GLOCC transformations (see Appendix
\ref{appGLOCC}) this implies that any GLOCC transformation that is possible
among pure fully entangled GFS can be implemented via GLUs. Hence, as before
we consider larger classes of operations, namely probabilistic transformations
and FLOCC' transformations. More precisely, we briefly explain how the GSLOCC
classes can be characterized in the multimode case for classes which contain a
critical state. We conclude this section by briefly discussing non-trivial
FLOCC' transformations among pure multimode GFS.

\subsection{Gaussian separable transformations}
We investigate Gaussian separable transformations (GSEP) among pure fully
entangled multimode states, i.e., multimode FS with the property that the
Schmidt decomposition (of the state in its Jordan-Wigner representation) with
respect to the splitting of one party versus the rest has no zero Schmidt
coefficients.
As stated in the following Lemma we show that such transformations are only
possible if the map corresponds to applying GLUs.

\begin{lemma} Let ${\cal E}_{sep}$ denote a Gaussian trace preserving
  separable map which transforms at least one pure fully entangled $m_1\times
  m_2\times \ldots\times m_N$--mode FS, $\ket{\Psi}$, into another pure fully
  entangled $m_1\times m_2\times \ldots\times m_N$--mode FS,
  $\ket{\Phi}$. Then, it holds that ${\cal E}_{sep}(\rho)=(U_1\tilde{\otimes}
  U_2\dots \tilde{\otimes} U_N) \rho (U_1^\dagger \tilde{\otimes}
  U_2^\dagger\dots \tilde{\otimes} U_N^\dagger)$ for all $\rho$.
\end{lemma}

Note that this lemma holds, as in the $n$--partite $n$--mode case for all FS (not only GFS).

\begin{proof}
  This lemma can be shown using an analogous argumentation as in the proof of
  Lemma \ref{LemmaGSEP}. We recall here the main steps of the proof and
  comment on its generalization to the multimode case. As argued in Appendix
  \ref{appGLOCC} Gaussian separable maps correspond to product operations,
  i.e., they are of the form
  $\cE_{sep}=\cE_1\tilde{\otimes}\cE_2\tilde{\otimes}\dots\tilde{\otimes}\cE_N$
  with GCPTMs $\cE_i$ which act now on $m_i$ modes. Analogously to the case of
  a single mode per site we consider $\rho=\id\tilde{\otimes}_{k\neq 1}{\cal
    E}_k(\proj{\Psi})$ with spectral decomposition $\sum_i p_i
  \proj{\Psi_i}$. As before it follows straightforwardly that for $p_i\neq 0$
\begin{eqnarray}
  \label{EQSEPLUmulti} \ket{\Phi}\propto A_k\tilde{\otimes}_j\id_{m_j}
  \ket{\Psi_i}\propto A_l\tilde{\otimes}_j\id_{m_j}  \ket{\Psi_i},\end{eqnarray}
where the operators $A_t$ are the Kraus operators of ${\cal E}_1$ and
$\id_{m_j}$ denotes here the identity on $m_j$ modes. We show next that there exists a Schmidt decomposition of the
Jordan-Wigner representation of $\ket{\Psi_i}$ in the splitting of the first $m_1$ modes
versus the remaining modes such that all involved local (with respect to that
splitting) states are fermionic. In order to do so note that the reduced state
of the first $m_1$ modes has to be fermionic and therefore the range of the
reduced state is spanned by FSs. Hence, any purification of this state (in particular $\ket{\Psi_i}$) is given by $\sum_{j=1}^{2^{\min
    (m_1,n_1)}}\lambda_j\ket{\eta_j}\ket{\nu_j}$, where $\ket{\eta_j}$
are orthogonal FSs of $m_1$ modes. That the $n_1 \equiv
\sum_{j=2}^N m_j$--mode states $\ket{\nu_j}$ are also fermionic, follows from the facts that the projector onto the
states $\ket{\eta_j}$ are fermionic operators (as they are sums of only even monomials in the Majorana operators) and that $\ket{\Psi_i}$ is a FS.
Moreover, as the
final state $\ket{\Phi}$ is fully entangled all Schmidt coefficients of
$\ket{\Psi_i}$ have to be unequal to zero (see Eq.\ (\ref{EQSEPLUmulti})),
i.e., $\lambda_j\neq 0$ $\forall j \in \{1, \ldots, 2^{\min (m_1,n_1)}\}$.\\
Analogous to the case of a single mode per site one can now apply
$\proj{\nu_j}$ on both sides of Eq.~(\ref{EQSEPLUmulti}) in order to see that
the action of $A_k$ on a basis is the same (up to a proportionality factor)
for all Kraus operators $A_k$ and hence $\cE_1$ is a Gaussian unitary
operation. Rearranging the modes \footnote{Note that exchanging the order of
  the parties (but keeping the relative order among the modes belonging to one
  party) neither changes the product structure of the maps due to Lemma
  \ref{le:KrausOpParity} nor the Kraus operators. Moreover, the Schmidt
  coefficients of the state in Jordan-Wigner representation are not changed by
  such a rearranging.} and applying the same argumentation for the various
parties proves the lemma.
\end{proof}
As GSEP is defined such that it includes all GLOCC transformations (see
Appendix~\ref{appGLOCC}) this lemma implies that non-trivial GLOCC
transformations among pure fully entangled GFS are not possible even if one considers the case
of an arbitrary (finite) number of modes per site. Hence, in the following
section we will consider probabilistic local transformations and comment on
the characterization of the GSLOCC classes for multimode states.

\subsection{Gaussian Stochastic LOCC}
As deterministic transformation are not possible among pure fully entangled GFS we will consider next
stochastic GLOCC operations. We distinguish between bipartite and multipartite
GFS, as in \cite{BoRe04b} a decomposition for bipartite states was
introduced. For multipartite states we show similar to the single-mode per
site case that stable states can be brought into a normal form.

\subsubsection{Bipartite case}
For bipartite pure multimode states, i.e., party A (B) holds $d_1$ ($d_2$)
modes respectively, it was shown in \cite{BoRe04b}  that one can consider
without loss of generality two subsystems of $d$ modes each, where $d=\min (d_1, d_2)$, that is the two parties hold the same number of modes. Thus, we only consider $d \times d$ states here.
A direct consequence of the results obtained in \cite{BoRe04b} is the following observation for bipartite multimode GSLOCC classes.
\begin{observation}For $d\times d$ modes (GFS) there exist $d$ different GSLOCC classes. \end{observation}

\begin{proof} This can be easily shown by using that any such state is up to
  GLU equivalent to $\otimes_{i=1}^{d} \ket{\Psi_i}_{A B}$, with
  $\ket{\Psi_i}_{A B} = \cos{\theta_i} \ket{00}_{A B} +\sin{\theta_i}
  \ket{11}_{A B}$ \cite{BoRe04b}.  Thus, A and B share $d$ 2-mode states
  $\ket{\Psi_i}_{A B}$, which are entangled for $\theta_i \neq 0,
  \pi/2$. Moreover, each GSLOCC class is characterized by the local rank of
  the states (the rank of the reduced states $\rho_{A}, \rho_B$ does not increase under GSLOCC) \cite{Nie99} and, hence, we immediately arrive at the
  above stated result. \end{proof}
Thus, there exist as many GSLOCC classes for bipartite GFS as SLOCC classes for bipartite qudit states.

\subsubsection{Multipartite case}

Analogously to the case of a single mode per site one can transform any multi-mode FS into a normal form by consecutively applying fermionic local invertible operators. Note again that this normal form vanishes for states in the null cone. Moreover, there exist semi-stable states that tend to a non-zero normal form but their SLOCC class does not contain a critical state \cite{GoWa13}.

\begin{lemma} \label{LemmaNormalFmultimode} Let $\ket{\Psi}$ be an entangled $m_1\times m_2\times \ldots\times m_N$--mode FS. Then $\ket{\Psi}$ can be constructively transformed (by applying invertible fermionic operators) into a unique (up to LUs) critical FS, $\ket{\Psi_s}$ (up to a proportionality factor which can tend to 0).\end{lemma}

The lemma can be proven by using the same argumentation as in the case of a
single mode per site (see Lemma\ \ref{LemmaNormalF}). Note that the only difference is that the local invertible operators, i.e., the reduced states, are no longer diagonal and thus, not automatically also Gaussian. However, they are general fermionic operators.
Note, furthermore, that any GSLOCC class containing a critical state can be easily characterized via this state. That is, if $\ket{\Psi_s}$ is a critical GFS then any other state $\ket{\Psi}$ in the same GSLOCC class is given by  $M_1 \otimes M_2 \ldots \otimes M_n \ket{\Psi_s}= \ket{\Psi}$. Here, the operators $M_i$ are Gaussian invertible operators.

\subsection{Fermionic LOCC}

Transformations among fully entangled multimode GFS via FLOCC'  \footnote{Analogous to the case of a single mode per party we define FLOCC' in the multimode case as the class of maps that can be implemented via FLOCC and whose Kraus operators are local fermionic operators.} work analogously to the $n$-mode $n$-partite case. Note, however, that in this setting there is an additional freedom when one considers transformations to not fully entangled states. Similar to the finite dimensional qudit case and contrary to the single-mode case it is possible to reduce the local rank of the parties via FLOCC', leaving still all parties entangled with each other.

\section{Conclusion}
\label{sec:conclusion}
We investigated the entanglement of GFS. For this purpose, we first derived a standard form of the CM for mixed $n$-mode $n$-partite GFS. Any CM can be brought into this standard form via GLU. As the standard form is unique, any two GFS are GLU-equivalent iff their CMs in standard form coincide. Furthermore, we showed that only two of the definitions of separable FS from \cite{BCW07} are reasonable for GFS. This is due to the fact that any separable state should have the property that also two copies of this state are again separable. For our derivations we used the definition of separability which declares a state separable if it is given by a convex combination of product states which commute with the local parity operator. According to this physically meaningful definition any separable state can be prepared locally. Using this definition we showed that for pure fully entangled $n$-mode $n$-partite as well as multimode GFS any GSEP is equivalent to a GLU. Thus, there exist no non-trivial GLOCC transformations among pure fully entangled GFS. Due to this fact we consider then the larger class of GSLOCC. With the help of a result on normal forms of states from \cite{FrankNormalfo} we also characterized the GSLOCC classes in the Jordan-Wigner representation and furthermore, explicitly derive them for few-mode systems. Then, we investigated the more general FLOCC', which contains in particular finitely-many rounds FLOCC (see Sec.~\ref{sec:FLOCC}), to obtain insights into the various entanglement properties of GFS and we show how to identify the MES of pure $n$-mode $n$-partite GFS under FLOCC'.  \par

Let us finally compare the fermionic case investigated here with the bosonic and the finite dimensional scenarios. In all three cases a computable condition for two ($n$--partite $n$--modes or $n$--qubit) states to be (G)LU--equivalent has been presented \cite{GiKr14, kraus2010}. Regarding the bosonic Gaussian case, we have that GSLOCC coincide with GLOCC transformations. This follows from the fact that any GSLOCC operation can be completed to a deterministic transformation. Moreover, there exist GLOCC transformation among pure bosonic Gaussian states which are not just GLU transformations (see e.g. \cite{GiKr14}). The MES for bosonic Gaussian states is not known, however, in \cite{GiKr14} a class of three--mode states has been identified which can reach states which cannot be reached from any symmetric three--mode state (including the GHZ and W states). Regarding the finite dimensional case, there exist (not surprisingly) more SLOCC classes than for GFS. Moreover, for Hilbert spaces composed of local Hilbert spaces of equal dimensions, it has been shown that almost all states are isolated, i.e., the state can neither be reached, nor transformed into any other (not LU--equivalent) state via LOCC \cite{Gour2016,sauerwein2017}. This resembles the fermionic case. However, as mentioned before, the reason for this to be true stems from the fact that almost no state possesses a local symmetry.

It would be interesting to investigate another physically relevant scenario by imposing a (global) particle-number
selection rule (as it is observed by elementary fermions in nature) on
the states considered and studying state transformations via
number-preserving local operations. Moreover, as in the qudit case, the transformations from a multipartite state, where each party holds more than a single mode (a single qubit)
to a state whose local rank is smaller might well allow (more) non-trivial transformations, respectively. Physically motivated, restricted set of states, such as FS or GFS, are ideally suited for this investigation, as it will be more trackable than the general qudit case. Moreover, this class of states is rich enough so that the results derived for them have the potential to lead also to new insight into state transformations among qudit states.

\begin{acknowledgments}
The research of KS and BK was funded by the Austrian Science Fund
(FWF) Grant No. Y535-N16. GG acknowledges support by the Spanish Ministerio de Econom\'{\i}a y Competitividad through the Project FIS2014-55987-P. CS acknowledges support by the Austrian Science Fund (FWF) Grant No. Y535-N16, the DFG and the ERC (Consolidator Grant 683107/TempoQ).
\end{acknowledgments}

\appendix
\section{}
\label{App:A}
In this appendix we study first the Choi-Jamiolkowski (CJ) isomorphism \cite{Cho72,Jam72,Cirac2000} among Gaussian states and Gaussian CP maps. Note that similar aspects of Gaussian CP maps have already been studied in \cite{Brav05}. However, there the author was using a different definition of the "tensor product" ($\otimes_f$) in the calculation. We summarize here the results using our notation. Then, we consider Gaussian LOCC (GLOCC) transformations and show that any GLOCC corresponds via the CJ isomorphism to a separable state. These investigations lead to the natural definition of fermionic separable maps (FSEP). Considering then the possible states which can be generated via GLOCC enables us to rule out the definition $\cS2_{\pi'}$ for separable states. That is, if $\cS2_{\pi'}$ does not coincide with $\cS2_{\pi}$ for GFS there exist states in $\cS2_{\pi'}$ which can neither be prepared locally by Gaussian operations, nor do they belong to the limit of such a preparation scheme.

\subsection{Choi-Jamiolkowski isomorphism in the Gaussian case}

The CJ isomorphism is a one to one mapping between CP maps and positive semidefinite operators. Denoting by ${\cal E}$ the CP map that is acting on $n$ modes and by $\rho_{\cal E}$ the corresponding operator we have 
\begin{eqnarray}
\rho_{\cal E}&=&{\cal E}\tilde{\otimes} \one (\kb{\Phi^+_{2n}}{\Phi^+_{2n}}) \nonumber \\
{\cal E}(\rho)&=&\tr_{2 3}(\rho_{\cal E}^{12} \rho^3 \ket{\Phi^+_{2n}}^{23}\bra{\Phi^+_{2n}}), 
\end{eqnarray}
where $\ket{\Phi^+_{2n}}\propto \prod_{a=1}^{2n} (\one +i \tilde{c}_a\tilde{c}_{2n+a})$. In \cite{Cirac2000} it has been shown that separable maps correspond to separable operations and that several other properties of the operators can be inferred from the maps and vice versa. The aim of this section is to show that the same isomorphism holds for Gaussian states. In the subsequent subsection we will then investigate the relation between separable operators and the corresponding maps. Note that we write Gaussian states and operators in this section in the Grassmann representation, see \cite{Brav05} for more details. Note further that $\rho_{\cal E}$ is a GFS iff ${\cal E}$ is a Gaussian map. It is obvious that $\rho_{\cal E}$ is a Gaussian state if ${\cal E}$ is Gaussian as $\ket{\Phi^+_{2n}}$ is a GFS.  Moreover, due to ${\cal E}(\rho)=\tr_{2 3}(\rho_{\cal E}^{12} \rho^3 \ket{\Phi^+_{2n}}^{23}\bra{\Phi^+_{2n}})$ one obtains that if $\rho_{\cal E}$ is a GFS then also ${\cal E}(\rho)$ is Gaussian for all GFS $\rho$ and therefore ${\cal E}$ corresponds to a Gaussian map.\\
In \cite{Brav05} it was shown that a linear CP map on $n$ fermionic modes is Gaussian iff it has
a (Grassmann) integral representation
\begin{equation}
  \label{eq:5}
  {\cal E}(X)(\theta) = C\int D\eta D\mu \exp[S(\theta,\eta)+i\eta^T\mu]X(\mu),
\end{equation}
where
\[S(\theta,\eta)=\frac{i}{2}
\left( \begin{array}{c} \theta\\ \eta\end{array} \right)^T
\left( \begin{array}{cc} A & B\\ -B^T& D
\end{array} \right)
\left( \begin{array}{c} \theta\\ \eta\end{array} \right)\equiv
\vec{\theta}^{\,T}M_{\cal E}\vec{\theta},
\]
with $C\geq 0$, real $2n\times 2n$ matrices $A,B,D$ and $M_{\cal E}^TM_{\cal E}\leq \one$. The identity map on
$n$ modes is given by $A=D=0$ and $B=\id$. Thus for a map ${\cal E}'$ on $n+m$
modes that acts non-trivially only on the first $n$ modes we take $A'=A\oplus0, D'=D\oplus0, B'=B\oplus\id$. Applying this map
(for $m=n$) to the maximally entangled state of $2n$ modes, we get as the CM
of the output state (with $\vec{\theta}=(\theta,\theta')$ (and same for
$\vec{\eta},\vec{\mu}$) and $\vec{x}_{12}=(\vec{\theta},\vec{\eta}), \vec{x}_{23}=(\vec{\eta},\vec{\mu})$)
\begin{align*}
  &\int D\eta D\eta' D\mu D\mu'
  e^{\frac{i}{2}\vec{x}_{12}^T\left( \begin{smallmatrix} A' & B'\\ -B'^T& D'
\end{smallmatrix}\right)\vec{x}_{12}+i\vec{\eta}^T\vec{\mu}}e^{\frac{i}{2}\vec{\mu}^T\left( \begin{smallmatrix}
  0^{\phantom{T}}& \id\\ -\id^{\phantom{T}}& 0\end{smallmatrix}\right)\vec{\mu}}\\
&=e^{\frac{i}{2}\vec{\theta}^T(A\oplus0)\vec{\theta}}\int
Dx_{23}e^{y^T\vec{x}_{23}+\frac{i}{2}\vec{x}_{23}^T\tilde{M}\vec{x}_{23}}\\
&\propto
e^{\frac{i}{2}\vec{\theta}^T(A\oplus0)\vec{\theta}}e^{-\frac{i}{2}y^T\tilde{M}^{-1}y}.
\end{align*}
In the last step we used the Gaussian integration rule (see Eq.~(13) of \cite{Brav05}),
 $y=(iB'^T\vec{\theta},0)$ and $\tilde{M}=\left( \begin{smallmatrix}
    D&0&\id&0\\0&0&0&\id\\-\id&0&0&\id\\ 0&-\id&-\id&0
\end{smallmatrix}\right)$. Since $y$ is non-zero only in the first two components,
we only need the upper diagonal block of the $2\times2$ blockmatrix $\tilde{M}^{-1}$, which
is given by the Schur complement as
$\left( \begin{smallmatrix} D & -\id\\ \id&0
\end{smallmatrix}\right)^{-1}=\left(\begin{smallmatrix}
0&\id\\ -\id& D\end{smallmatrix} \right)$.
Thus, we end up with a Gaussian Grassmann representation with CM
\begin{equation}
  \label{eq:6}
  \left(\begin{array}{cc} A &B\\ -B^T&D \end{array} \right).
\end{equation}
Hence, the GFS with this CM is the CJ-state $\rho_{\cal E}$ of the map $\cal E$. Note that by using the above mentioned definition of a tensor product $\otimes_f$ (see Def.~5 in \cite{Brav05}) for the computation of the CJ-state we obtain a CM $\left( \begin{smallmatrix} A &-B\\ B^T&D \end{smallmatrix} \right)$. The corresponding state is obtained by applying the local operator $\prod_{i=1}^{2n} \tilde{c}_i$ to $\rho_{\cal E}$.

In order to confirm that the state $\rho_{\cal E}$ with CM given in Eq.~\eqref{eq:6} allows for the physical interpretation, which is characteristic for the CJ-state, that it can be used to realize the map $\cal E$ via teleportation, we compute
\[
\tr_{23}(\rho_{\cal E}^{12}\rho_\Gamma^3\ket{\Phi^+_{2n}}^{23}\bra{\Phi^+_{2n}})).
\]
Here, the superscripts indicate on which of the three different blocks of
modes the state is nontrivial. Using the formula for the trace of two
operators X,Y in Grassmann variables \footnote{The trace is given by
  \unexpanded{$\tr(X Y) = (-2)^n \int D \theta D \mu e^{\theta^T \mu} w(X,
    \theta) w(Y,\mu)$} with \unexpanded{$w(\tilde{c}_p
    \tilde{c}_q...\tilde{c}_r, \theta) = \theta_p \theta_q ... \theta_r$}.}
(see also Eq.~(15) in \cite{Brav05}) and with $X=\rho_{\cal
  E}^{12}\rho_\Gamma^3$, $Y=\ket{\Phi^+_{2n}}^{23}\bra{\Phi^+_{2n}})$ the
trace is given by

\begin{align*}
 & \tr_{23}(X Y) \propto \\ & \int D\vec{\eta} D\vec{\mu} e^{
    (iB^T\theta)^T\eta + \frac{i}{2}\left(\theta^TA\theta
+ \eta^TD\eta +\eta'^T\Gamma\eta'+\vec{\mu}^T\left( \begin{smallmatrix} 0&\id\\ -\id&0
\end{smallmatrix} \right)\vec{\mu}\right)}e^{\vec{\eta}^T\vec{\mu}} \\
&= e^{ \frac{i}{2}\theta^TA\theta} \int D\vec{x}_{23}e^{\xi^T\vec{x}_{23} + \frac{i}{2}\vec{x}_{23}^TM^\prime\vec{x}_{23}}.
\end{align*}

Here, again $\vec{x}_{23} = (\vec{\eta},\vec{\mu})$ and
\begin{align*}
  \xi^T&= ((iB^T\theta)^T,0,0,0),\\
M^\prime&=\left( \begin{array}{cccc} D & 0 & -i\id & 0\\
0 & \Gamma & 0 &- i\id\\
i\id & 0 & 0 & \id\\
0 & i\id & -\id & 0
\end{array} \right).
\end{align*}
Using again the Gaussian integration rule (Eq.~(13) in \cite{Brav05}) we obtain as a result a Gaussian state with CM
\begin{align}
  \Gamma_{\mathrm{out}} &=
A-(iB)\left(\left[\left( \begin{array}{cc} D & \\ &\Gamma
\end{array} \right)-\left( \begin{array}{cc} 0&\id\\-\id&0
\end{array} \right)^{-1}  \right]^{-1}\right)_{11}(iB^T) \nonumber \\
&=
A+B\left(\left( \begin{array}{cc} D & \id\\ -\id&\Gamma
\end{array} \right)^{-1}\right)_{11}B^T \nonumber \\
&= A+ B\left(D+\Gamma^{-1} \right)^{-1}B^T,
\label{gammaout}
\end{align}
which is just ${\cal E}(\rho_{\Gamma})$.

Summarizing, we have shown that the state $\rho_{\cal E}=({\cal E}\tilde{\otimes}\id)(\ket{\Phi^+_{2n}}\bra{\Phi^+_{2n}})=\rho_{M}, $ where the GFS $\rho_M$ with CM $M=\left( \begin{smallmatrix}
      A&B\\ -B^T&D
\end{smallmatrix} \right)$ is the CJ-state of the Gaussian map $\cal E$ given in Eq.~\eqref{eq:5} or equivalently as the Gaussian map which maps the CM $\Gamma$ to $\Gamma_{\mathrm{out}}$ as given in Eq.~\eqref{gammaout}.

\subsection{Gaussian LOCC (GLOCC)}\label{appGLOCC}
Let us now investigate the relation of the entanglement properties of CJ-state
and the entanglement properties of the corresponding CP map. We will consider
here only bipartite systems, however, all arguments hold also for the
multipartite setting. In case of finite dimensional systems a CPTM, ${\cal
  E}$, is called separable if it can be written
as \begin{equation} \label{Eq_Sep} {\cal E}(\rho)=\sum_k A_k \otimes B_k \rho
  A^\dagger_k \otimes B^\dagger_k.\end{equation} As the set of separable maps
(SEP) strictly contains the set of LOCC, i.e., the set of maps which can be
realized via local operations and classical communication, SEP lacks a clear
physical meaning. Hence, when considering restricted sets of maps, such as
here fermionic or Gaussian maps, there is no clear way of
specializing the notion of SEP to these sets. This is why we consider here the physically
meaningful, however, mathematically generically much less tractable set of
LOCC, for which this specialization is obvious. We will then show that
this consideration suggests the natural definition of fermionic SEP (FSEP). 

Let us first consider the CJ-state of a local map, ${\cal E}={\cal E}_{1}\tilde{\otimes}{\cal E}_2$, i.e., a composition of two maps, ${\cal E}_{1}$, and ${\cal E}_{2}$, which act on the first and second system nontrivially, respectively. In this case, the CM of
the CJ-state splits in the form $A=A_1\oplus A_2, B=B_1\oplus B_2, D=D_1\oplus
D_2$. One can easily check that ${\cal
  E}(\Phi_+)$ is separable with respect to the splitting $13|24$ according to our
definition (see Sec.~\ref{sec:mixed-state-ferm}). Hence, using our definition of separability ($\cS2_{\pi}$), the CJ isomorphism maps local maps to separable states.

Let us next show that the CJ-state of any Gaussian LOCC is separable according to the definition $\cS2_{\pi}$. That is, we show that
any map which describes a GLOCC corresponds to a Gaussian CJ-state whose CM is given by $\Gamma_1\oplus \Gamma_2$ \footnote{It can be easily seen that a CM of this form can be obtained by rearranging the modes if a state has a CM  \unexpanded{$A=A_1\oplus A_2, B=B_1\oplus B_2, D=D_1\oplus
D_2$}.}, i.e., the corresponding state factorizes. Using this result and the remark above, it is then easy to see that any GLOCC can be written as ${\cal E}_{1}\tilde{\otimes}{\cal E}_2$.

Operationally, a finitely-many-rounds FLOCC protocol is a protocol which can
be realized by local fermionic operations and a finite number of rounds of
classical communication. In order to include also FLOCC protocols which
require infinitely many rounds of communication we define FLOCC as the set of
finitely-many-rounds FLOCC protocols together with those, which are the limit
of a sequence of such protocols. A Gaussian FLOCC is a FLOCC that can be
implemented with Gaussian means and that maps Gaussian states to Gaussian
states. Stated differently any map, ${\cal E}$, corresponding to a
finitely-many-rounds FLOCC (GLOCC) can be written as in Eq.~(\ref{Eq_Sep}),
where all operators, $A_k,B_k$ are fermionic (Gaussian) operators,
respectively. Any map within FLOCC (GLOCC) can be written as the limit of a
sequence of such maps where each element of the sequence is obtained by
applying one more round of a FLOCC protocol to the preceding element.

Let us now show that the CJ-state of a Gaussian FLOCC factorizes. We consider first finitely many round protocols and extend the result then to the limit of sequences of such protocols. The CJ-state is given by $\rho_{\cal E}=({\cal E}^{a b}_{\mathrm{FLOCC}}\tilde{\otimes} \one^{a^\prime b^\prime} )(P_{\Phi^+}^{a a^\prime} \tilde{\otimes} P_{\Phi^+}^{b b^\prime})$. Using that ${\cal E}_{\mathrm{FLOCC}}$ is of the form given in Eq.~\eqref{Eq_Sep}, where $A_k, B_k$ are fermionic operators and computing the expectation value of $\tilde{c}_a \tilde{c}_b$, where $\tilde{c}_a$ $(\tilde{c}_b)$ denotes any Majorana operator acting on modes in $a$ ($b$), respectively, we obtain
\bea &\tr[{\cal E}^{a b}_{\mathrm{FLOCC}}\tilde{\otimes} \one^{a^\prime b^\prime} (P_{\Phi^+}^{a a^\prime} \tilde{\otimes} P_{\Phi^+}^{b b^\prime}) \tilde{c}_a \tilde{c}_b]= \\ \nonumber
 &\sum_k (-1)^{f(A_kB_k)}\tr[A^\dagger_k \tilde{c}_a A_k \tilde{\otimes} B^\dagger_k \tilde{c}_b B_k (P_{\Phi^+}^{a a^\prime} \tilde{\otimes} P_{\Phi^+}^{b b^\prime})],  \label{expVal} \eea
where $f(A_k B_k)=0$ for even operators $A_k B_k$ and $f(A_k B_k)=1$ for odd operators.
As any operator $A_k$ is parity respecting, i.e., is fermionic, $A^\dagger_k
\tilde{c}_a A_k$ is an odd operator. Due to the fact that the projector onto
$\Phi^+$ is even and that $\tr\left[
    A\tilde{\otimes}B (P_{\Phi^+}^{aa'}\tilde{\otimes} P_{\Phi^+}^{bb'})\right]=\tr\left[
    AP_{\Phi^+}^{aa'}\right]\tr\left[BP_{\Phi^+}^{bb'}\right]$, the trace vanishes. Hence, the off--diagonal terms in the CM
of the CJ-state vanish and $\Gamma=\Gamma^{aa^\prime}\oplus
  \Gamma^{bb^\prime}$. In case that the CJ-state is Gaussian, in particular,
if ${\cal E}^{a b}_{\mathrm{FLOCC}}$ is a Gaussian map \footnote{Let us remark
  here that the CJ-state of a Gaussian map is, of course, Gaussian as $\Phi^+$
  is Gaussian.}, we hence have that the CJ-state factorizes. The last assertion
follows from the fact that for GFS, Wick's theorem holds and thus, all
higher-order correlations factorize if the CM is
  block-diagonal. In case ${\cal E}^{a b}_{\mathrm{FLOCC}}$ is the limit of a sequence of finitely-many-rounds protocols, the statement also holds due to continuity arguments \footnote{In \cite{whatyouLOCC} instrument convergence (see also Sec.\ \ref{sec:FLOCC}) was shown by using the distance measure induced by the diamond norm of the corresponding CPTMs, i.e., $\parallel {\cal E} - \tilde{{\cal E}} \parallel _\diamond$. That is, for all $ {\cal E}$, that are the limit of a sequence of finitely many round protocols, there exists a finitely many round protocol $ \tilde{{\cal E}}_n$ such that $\lim_{n \rightarrow \infty}\parallel {\cal E} - \tilde{{\cal E}}_n \parallel _\diamond = 0$. This implies convergence of the corresponding CJ-states $\rho_{\cal E},\rho_{\tilde{\cal E}_n}$ in trace norm, i.e., $\lim_{n \rightarrow \infty} \parallel \rho_{\cal E} -\rho_{\tilde{\cal E}_n} \parallel _1 = 0$. This leads to the continuity of the expectation value of $\tilde{c}_a \tilde{c}_b$ in Eq.~(\ref{expVal}), i.e., $\lim_{n \rightarrow \infty}  \tr[  (\rho_{\cal E}-\rho_{\tilde{\cal E}_n}) \tilde{c}_a \tilde{c}_b] =0$.
}.

It is evident from the discussion above that (i) any FLOCC applied to a
product state is separable and that (ii) any separable GFS (according to
$\cS2_{\pi}$) can be generated via FLOCC from a product state. This fact,
being obvious from a physical point of view, shows that the definition we
choose for separability meets the necessary requirements. Moreover, this also
shows that states which are convex combination of non--fermionic states (or
the limit thereof) and for which no decomposition into FSs exist cannot be
generated locally. Hence, in case the set $\cS2_{\pi'}$  contains such a
state, then calling states in $\cS2_{\pi'}$  separable does not conform to the usual operational definition.

Note that in the argument above the restriction to locally realizable maps has never been used. Hence, a natural definition of Gaussian separable maps (GSEP) is the set of CPTMs whose CJ-state is a separable Gaussian state, i.e., $\rho_{{\cal E}_{\mathrm{GSEP}}}=\rho_A\tilde{\otimes} \rho_B$ (which for GFS is equivalent to $\Gamma_{{\cal E}_{\mathrm{GSEP}}}=\Gamma_A\oplus \Gamma_B$). Note that this implies that ${\cal E}_{\mathrm{GSEP}}={\cal E}_A\tilde{\otimes} {\cal E}_B$. FSEP is then defined as the set of
CPT maps that can be written as $\cE(\rho)=\sum_k (A_k\tilde{\otimes}
B_k)\rho(A_k\tilde{\otimes} B_k)^\dag$ where all the $A_k,B_k$ are parity-respecting
operators.

 \section{Proof of Observation 2}
 \label{app:proofObs2}
 Here, we prove the observation that a product state according to definition  $\cP1_{\pi}$, i.e., the set of states for which the expectation values of all products of physical observables factorize, can have non-zero correlation between A and B. \begin{proof}
Let us denote by $\cP1_{\pi}$ the set of
states for which all products of locally measurable
observables factorize, by $\cS1_{\pi}$ its convex hull, and by $\cS_G$ the
set of Gaussian states. We show that $\rho\in \cS1_{\pi}\cap\cS_G$ implies $\Gamma_\rho=\left( \begin{array}{cc}
\Gamma_A&C\\-C^T&\Gamma_B\end{array} \right)$ with $\rank \,C\leq1$ and that
there are such states with $\rank \,C=1$.\\
We consider observables of the form $\Pi_{i=1}^{2n_a}\tilde{c}^a_{k_i}$ and
$\Pi_{j=1}^{2m_b}\tilde{c}^b_{l_j}$, where $\tilde{c}^{a(b)}$ refer to Majorana operators on
Alice's (Bob's) modes. We exploit the fact that we can compute their
expectation values in two ways: either by using the Wick formula for the
$n+m$-mode Gaussian state or by using the separability condition and using the
Wick formula twice for the $n$ and $m$ local modes separately. We show that
these only coincide for all observables if the rank of the off-diagonal block
$C$ of the full CM is not larger than 1.\\
Considering the observable
$\tilde{c}^{a}_{k_1}\tilde{c}^{a}_{k_2}\tilde{c}^b_{l_1}\tilde{c}^b_{l_2}$, we find that
$C_{k_1l_2}C_{k_2l_1}=C_{k_1l_1}C_{k_2l_2}$ where $C=(C_{ij})_{ij}$. W.l.o.g. we can choose to work in the
basis in which $C$ takes diagonal form
(i.e., apply local basis changes $O_a,O_b$ such that $O_aCO_b^T$ is diagonal
(singular value decomposition)). Then, considering $k_1=l_1,k_2=l_2$ one obtains that the
rank of $C$ can be at most one since two non-zero singular values would lead
to a contradiction.
This single non-zero entry, however, can not lead to any difference between
the two ways of computing expectation values of products of even observables
and thus there can be (and are) Gaussian states in $\cS_{S1}$ with $C\not=0$:
For example, consider any Gaussian state with CM such that
  $C_{k_1k_1}\not=0$ is the only non-zero entry of $C$ and consider any pair of even observables 
$A=\Pi_i\tilde{c}^a_{k_i}, B=\Pi_j\tilde{c}^b_{l_j}$, then
$\rho_\Gamma(AB)=\rho_{\Gamma_A\oplus\Gamma_B}(AB)=\rho_{\Gamma_A}(A)\rho_{\Gamma_B}(B)=\rho_\Gamma(A)\rho_\Gamma(B)$,
since, using Wick's formula any term that contains a pairing
$(k_1,k_1)$ must necessarily contain another AB-correlating pair $(k_2,l_2)$
with $k_1\not=k_2,l_1\not=l_2$ since no index appears twice in the same
  subsystem. However, since $C_{k_1k_1}$ is the
  only non-vanishing entry of $C$ the corresponding term is zero and only the local
blocks $\Gamma_A,\Gamma_B$ contribute to $\rho(AB)$.
\end{proof}
\section{Standard form of the CM of $1 \times 1 \times 1$ states}
\label{3modeGLU}
Here, we state the conditions on the parameters of the standard form for mixed 3 modes GFS, i.e.,
  \bea\nonumber
\hspace*{-0.2cm}  S(\gamma)\hspace*{-0.1cm} = \hspace*{-0.1cm} \left(\small{ \begin{array}{cccccc} 0 & \lambda_1&d_{12}&0&l_1 d_{13}&l_2 d^\prime_{13}\\
-\lambda_1 &0&0&d^\prime_{12}&-l_2 d_{13}&l_1 d^\prime_{13}\\
-d_{12}&0&0&\lambda_2&m_1&m_{12}\\
0&-d^\prime_{12}&-\lambda_2&0&m_{21}&m_2\\
 -l_1 d_{13}&l_2 d_{13}&-m_1&-m_{21}&0&\lambda_3\\
-l_2 d^\prime_{13}&-l_1 d^\prime_{13}&-m_{12}&-m_2&-\lambda_3&0
\end{array} }\right).
\eea
in more detail.
 If no mode factorizes we have for $\lambda_i> 0$ for $i\in\{1,2,3\}$ the following cases:
\begin{itemize}
\item  $d_{12}>|d^\prime_{12}|$ and
\begin{itemize}
\item $d_{13}>|d^\prime_{13}|$ and $l_1^2+l_2^2=1$ with either $l_1>0$ or
$l_1=0$ and $l_2>0$ or
\item  $d_{13}=|d^\prime_{13}| \neq 0$, $l_1=1$ and $l_2=0$ or
\item $l_1=l_2=0$, $m_1=
l^\prime_1 d_{23}$, $m_2=l^\prime_1 d^\prime_{23}$, $m_{12}=l^\prime_2 d^\prime_{23}$ and $m_{21}=-l^{\prime}_2 d_{23}$ with  $l_1^{\prime 2}+l_2^{\prime 2}=1$, $d_{23}>|d^\prime_{23}|$ and either $l_1'>0$ or $l_1'=0$ and  $l^\prime_2> 0$ or
\item  $l_1=l_2=0$, $m_1=|m_2|\neq 0$, $m_{12}=0$ and $m_{21}=0$.
\end{itemize}
\item $d_{12}=|d^\prime_{12}| \neq 0$ and
\begin{itemize}
\item $d_{13}>|d^\prime_{13}|$, $l_1=1$ and $l_2=0$ or
\item$d_{13}=|d^\prime_{13}| \neq 0$, $l_1=1$, $l_2=0$, $m_1=
l^\prime_1 d_{23}$, $m_2=l^\prime_1 d^\prime_{23}$, $m_{12}=l^\prime_2 d_{23}$ and $m_{21}=-l^{\prime}_2 d^\prime_{23}$ with $d_{23}>|d^\prime_{23}|$ and $l_1^{\prime 2}+l_2^{\prime 2}=1$ or
\item $d_{13}=|d^\prime_{13}| \neq 0$, $l_1=1$, $l_2=0$, $\gamma_{23}\propto O(2,\R)$ and $d^\prime_{12}d^\prime_{13}|\gamma_{23}|> 0$ or
\item $d_{13}=|d^\prime_{13}| \neq 0$, $l_1=1$, $l_2=0$, $m_1=|m_2|$, $m_{12}=m_{21}=0$ and  $d^\prime_{12}d^\prime_{13}m_2< 0$ or
\item  $l_1=l_2=0$, $m_1> |m_2|$, $m_{12}=0$ and $m_{21}=0$  or
\item  $l_1=l_2=0$, $m_1=|m_2|\neq 0$, $m_{12}=0$ and $m_{21}=0$  or
\item $d_{13}=|d^\prime_{13}| \neq 0$, $l_1=1$, $l_2=0$ and $m_1=
m_2=m_{12}=m_{21}=0$.
\end{itemize}
\item $d_{12}=|d^\prime_{12}|=0$ and
\begin{itemize}
\item   $d_{13}> |d^\prime_{13}|$, $l_1=1$, $l_2=0$, $m_1=
l^\prime_1 d_{23}$, $m_2=l^\prime_1 d^\prime_{23}$, $m_{12}=l^\prime_2 d_{23}$ and $m_{21}=-l^{\prime}_2 d^\prime_{23}$ with  $l_1^{\prime 2}+l_2^{\prime 2}=1$, $d_{23}>|d^\prime_{23}|$ and either $l_1'>0$ or $l_1'=0$ and  $l^\prime_2> 0$ or
\item  $d_{13}> |d^\prime_{13}|$, $l_1=1$, $l_2=0$, $m_1= |m_2|\neq 0$, $m_{12}=0$ and $m_{21}=0$ or
\item $d_{13}= |d^\prime_{13}| \neq 0$, $l_1=1$, $l_2=0$, $m_1> |m_2|$, $m_{12}=0$ and $m_{21}=0$  or
\item $d_{13}= |d^\prime_{13}| \neq 0$, $l_1=1$, $l_2=0$, $m_1= |m_2|\neq 0$, $m_{12}=0$ and $m_{21}=0$.
\end{itemize}
\end{itemize}
In case $\lambda_i = 0$ for some $i\in\{1,2,3\}$ the standard form can be obtained analogously. However, in this case $m_i$ is not determined by $\gamma_{ii}$.

\bibliography{GFSreferences}

\end{document}